\documentclass[letter,11pt]{article}
\usepackage[latin9]{inputenc}
\usepackage{geometry}
\usepackage{color}
\usepackage{amsmath}
\usepackage{amsthm}
\usepackage{amssymb}
\usepackage{graphicx}
\usepackage{setspace}
\usepackage{esint}
\usepackage{amsfonts}
\usepackage{comment}
\usepackage{booktabs}
\usepackage{lscape}
\usepackage{chngpage}
\usepackage{array}
\usepackage[titletoc,title]{appendix}
\usepackage[labelsep=period]{caption}
\usepackage{epstopdf}
\usepackage{natbib}
\usepackage[colorlinks,citecolor=blue,urlcolor=blue,filecolor=blue]{hyperref}
\usepackage{diagbox}
\usepackage{mathtools}
\usepackage{threeparttable}
\usepackage{multirow}
\usepackage{float}
\usepackage{bm}
\usepackage{siunitx}
\usepackage{enumerate}
\usepackage{threeparttable}
\usepackage{enumitem}
\usepackage{tikz}
\usetikzlibrary{arrows.meta}
\providecommand{\U}[1]{\protect\rule{.1in}{.1in}}
\providecommand{\U}[1]{\protect\rule{.1in}{.1in}}
\providecommand{\U}[1]{\protect\rule{.1in}{.1in}}

\makeatletter
\theoremstyle{plain}

\theoremstyle{definition}

\theoremstyle{plain}

\RequirePackage[colorlinks,citecolor=blue,urlcolor=blue]{hyperref}
\RequirePackage{hypernat}
\providecommand{\U}[1]{\protect\rule{.1in}{.1in}}
\newtheorem{lemma}{Lemma}
\newtheorem{theorem}{Theorem}

\newtheorem{assumption}{Assumption}
\newtheorem{remark}{Remark}

\providecommand{\definitionname}{Definition}
\providecommand{\propositionname}{Proposition}
\providecommand{\theoremname}{Theorem}
\providecommand{\definitionname}{Definition}
\providecommand{\propositionname}{Proposition}
\providecommand{\theoremname}{Theorem}
\providecommand{\definitionname}{Definition}
\providecommand{\propositionname}{Proposition}
\providecommand{\theoremname}{Theorem}
\makeatother
\providecommand{\definitionname}{Definition}
\providecommand{\propositionname}{Proposition}
\providecommand{\theoremname}{Theorem}

\graphicspath{{figures/}}
\usepackage{physics} 
\usepackage{booktabs}
\usepackage{caption}
\usepackage{makecell}
\usepackage{pifont}
\usepackage{natbib}
 \bibpunct[, ]{(}{)}{,}{a}{}{,}
 \def\bibfont{\small}

\DeclareMathOperator{\pr}{\mathbb P}
\DeclareMathOperator{\E}{\mathbb E}
\DeclareMathOperator{\Var}{\mathbb{V}}
\DeclareMathOperator{\Cov}{\mathsf{Cov}}
\DeclareMathOperator{\cv}{\rm CoVaR}
\DeclareMathOperator{\R}{\mathbb R}
\DeclareMathOperator{\cvv}{ {\rm CoVaR}_{\alpha,\beta} }
\DeclarePairedDelimiter\ceil{\lceil}{\rceil}

\DeclareMathOperator{\D}{\mathbb D}

\DeclareMathOperator{\N}{\mathbb N}
\DeclareMathOperator{\Normal}{\mathsf{Normal}}
\usetikzlibrary{calc}

\begin{document}

\title{\bf Estimating Systemic Risk within Financial Networks:\\A Two-Step Nonparametric Method}

\author{
Weihuan Huang\\[-6pt]
{\small School of Management \& Engineering, Nanjing University, Nanjing 210093, China}\\[-8pt]
{\small hwh@nju.edu.cn}
}

\renewcommand{\thefootnote}{\fnsymbol{footnote}}
\footnotetext[1]{The author acknowledges the financial support provided by the National Natural Science Foundation of China (NSFC) Nos. 12301601, 72161160340, 72091211, 72122008, 72371125.}

\date{}
\maketitle

\begin{abstract}
\noindent
CoVaR (conditional value-at-risk) is a crucial measure for assessing financial systemic risk, which is defined as a conditional quantile of a random variable, conditioned on other random variables reaching specific quantiles. It enables the measurement of risk associated with a particular node in financial networks, taking into account the simultaneous influence of risks from multiple correlated nodes. However, estimating CoVaR presents challenges due to the unobservability of the multivariate-quantiles condition. To address the challenges, we propose a two-step nonparametric estimation approach based on Monte-Carlo simulation data. In the first step, we estimate the unobservable multivariate-quantiles using order statistics. In the second step, we employ a kernel method to estimate the conditional quantile conditional on the order statistics. We establish the consistency and asymptotic normality of the two-step estimator, along with a bandwidth selection method. The results demonstrate that, under a mild restriction on the bandwidth, the estimation error arising from the first step can be ignored. Consequently, the asymptotic results depend solely on the estimation error of the second step, as if the multivariate-quantiles in the condition were observable. Numerical experiments demonstrate the favorable performance of the two-step estimator.
\bigskip

\end{abstract}

\newpage{}

\onehalfspacing
\section{Introduction}\label{intro}

The global financial crisis of 2007-2009 is commonly attributed to the insufficient attention given to systemic risks. The recurring financial turmoils witnessed in the past decade have further reaffirmed this perspective: The financial system constitutes an  interconnected network where risk contagion often occurs from particular nodes, which, in the absence of proper evaluation and effective regulation, evolves into significant systemic risks. A typical financial network encompasses diverse entities, including financial institutions, firms, and individuals, with linkages established through mutual holdings of shares, debts, and other obligations. This interconnectedness within the network engenders risk contagion and presents formidable challenges to risk management. For instance, during the 2007-2009 financial crisis, the collapse of the U.S. housing market resulted in a substantial decline in the market value of related financial products, impacting the asset values of institutions holding these products. The risk was further amplified by high-risk behaviors and interconnections among financial institutions, ultimately culminating in the global crisis.  Therefore, managing systemic risk necessitates not only evaluating risks associated with individual nodes in financial networks but also evaluating the interdependencies between nodes. For further insights into financial networks, we refer to  \cite{elliott2014financial,allen2009networks,acemoglu2015systemic,babus2016formation,eisenberg2001systemic,jackson2021systemic,glasserman2015likely,glasserman2016contagion}, among others. 

The measurement of systemic risk stands as the most fundamental issue in risk management. The commonly employed measure by financial institutions is the value-at-risk (VaR), which quantifies an institution's portfolio loss at a specified \textit{quantile} \citep{Jorion}. However, VaR fails to account for the interconnections within financial networks, rendering it an inadequate measure of systemic risk. Expanding upon the concept of \textit{conditional quantile}, \cite{CoVaR2016} introduce CoVaR as a promising measure for evaluating financial systemic risk. Specifically, CoVaR represents the conditional quantile of an institution's loss, conditioned on another institution's loss reaching a specific VaR threshold. Based on CoVaR, they further introduce $\Delta\cv$ (difference of CoVaR) to represent the increase in VaR of an institution' loss when another institution's loss shifts from a median state to a crisis state. The measure $\Delta\cv$ effectively captures the extent to which the risk of an institution is affected by the risk of another institution. Through empirical analysis, they successfully demonstrate the predictability of the 2007-2009 financial crisis using $\Delta\cv$. Consequently, CoVaR has emerged as one of the most significant measures of financial systemic risk. For further insight into the concepts of quantile and conditional quantile, we refer readers to \cite{Serfling1980,koenker2005quantile}. 

However, CoVaR solely assesses the node-to-node systemic risk, thereby failing to capture the systemic risk within financial networks when multiple nodes simultaneously impact a specific node. Specifically, CoVaR measures the tail dependence between arbitrary pairs of random variables representing the losses of two respective institutions. Nevertheless, during a crisis, it is not uncommon to observe simultaneous crises occurring across numerous financial institutions. Consequently, it becomes necessary to consider the correlation of a specific node with multiple other nodes. This multilateral nature of financial networks prompts us to consider a measure of systemic risk associated with a specific node in financial networks, accounting for the simultaneous influence of risks stemming from multiple correlated nodes. In this paper, we study an extension of CoVaR from the univariate-quantile condition to a multivariate-quantiles condition. Specifically, CoVaR is defined as a conditional quantile of an institution's loss, conditioned on the losses of some correlated institutions reaching respective quantiles. Building upon this extension of CoVaR, $\Delta\cv$ with the multivariate-quantiles condition can be introduced to capture the systemic risk within financial networks. 

This extension is conceptually natural but presents challenges in statistical estimation due to the \textit{unobservability} of the multivariate-quantiles condition. It is important to note that the unobservability has two implications: Firstly, the condition represents an event with zero probability that cannot be directly observed; Secondly, the multivariate-quantiles within the condition are also unobservable. To address the challenges, we develop a two-step nonparametric approach to estimate CoVaR with a multivariate-quantiles condition using Monte-Carlo simulation data. Specifically, in the first step, we employ the order statistics to estimate the multivariate-quantiles in the condition. Note that, the order statistics is a conventional nonparametric method of quantile estimation \citep{Serfling1980}. In the second step, we develop a kernel estimation of the conditional quantile conditioned on the order statistics. Note that, the kernel method is a conventional nonparametric method of conditional quantiles estimation \citep{li2007nonparametric}. The distinction here is that the conventional conditional quantile assumes a fixed and known condition, whereas the multivariate-quantiles condition of CoVaR is unobservable. Consequently, in this paper, we combine these two nonparametric methods to estimate CoVaR with the unobservable condition. We underscore the advantages of employing Monte-Carlo simulation and the nonparametric methods, which are delineated as follows. 

\textit{Why use Monte-Carlo simulation?}---Real-world financial data often exhibits intricate data characteristics, such as temporal correlations, challenging the assumption of independent and identically distributed (i.i.d.) data. Nevertheless, we can leverage real financial data to calibrate a predetermined statistical model and subsequently employ this validated model to generate a significant volume of i.i.d. simulation data. This approach facilitates the acquisition of efficient and reliable CoVaR estimates through the utilization of Monte-Carlo simulation. On the other hand, compared to existing model-based methods that necessitate model simplification for analytical CoVaR computation, simulation-based methods offer greater modeling flexibility. The concept of utilizing simulation-based methods for the estimation of CoVaR with a univariate-quantile condition is first introduced by \cite{huanghong23}.

\textit{Why use nonparametric methods?}---Nonparametric methods are employed due to their ability to address the unobservability of the multivariate-quantiles condition through estimating the condition using order statistics (i.e., the first step) and estimating the conditional quantile using the kernel method (i.e., the second step). Compared to existing simulation-based methods \citep{huanghong23}, the kernel method enables the handling of conditional quantiles with a multivariate condition, making it suitable for the estimation problem of CoVaR within financial networks.

With the new two-step estimator, we engage in a thorough investigation of its asymptotic properties. We establish the consistency and asymptotic normality of the two-step estimator, accompanied by a bandwidth selection method. In particular, the asymptotic normality reveals that the rate of convergence of the two-step estimator is $(n\Delta_n)^{-1/2}$, where $n$ represents the sample size, $\Delta_n=\prod_{j=1}^m\delta_{n,j}$, $m$ denotes the number of dimensions in the condition, and $\delta_{n,j}$ represents the bandwidth parameters. Notably, in the special case of a univariate-quantile scenario (i.e., $m=1$), the two-step estimator can achieves a faster rate of convergence ($n^{-2/5}$) compared to the batching estimator ($n^{-1/3}$) provided in \cite{huanghong23}. During the proof, the estimation error of the two-step estimator is contingent upon the errors generated during each of the two respective steps. To address these errors comprehensively, we simultaneously handle both steps. Interestingly, we discover that, subject to a mild restriction on the bandwidth, the estimation error stemming from the first step can be disregarded. Consequently, the asymptotic results hinge solely on the estimation error of the second step and have the same form as those of the conventional conditional quantile, as if the multivariate-quantiles in the condition were directly observable.

Finally, we present two illustrative numerical examples within the framework of the delta-gamma approximation model. The first example involves a univariate-quantile condition for which an analytical form of CoVaR exists. The numerical results not only validate our theoretical asymptotic findings but also demonstrate the superiority of our proposed two-step estimator over the batching estimator presented in \cite{huanghong23}. In the second example, we consider a multivariate-quantiles condition for which an analytical solution of CoVaR is not available. The numerical results highlight the substantial disparity between the CoVaR with the multivariate-quantiles condition and the CoVaR derived from a univariate-quantile condition. Consequently, this underscores the importance of considering multiple institutions facing risk simultaneously and avoiding reliance on the CoVaR with a univariate-quantile condition, as it may lead to a flawed measure of systemic risk in practical applications. 

\subsection{Literature Review}

There is a body of literature focusing on the estimation method of CoVaR. One line of research employs \textit{model-based methods} for CoVaR, assuming specific structural models for the loss distribution. These models allow for estimation of the model using real financial data, enabling the calculation of CoVaR based on these estimated models. Building upon a linear factor model where the two institutions' losses exhibit a linear relationship, \cite{CoVaR2016} propose a quantile regression approach to estimate CoVaR. Moreover, both \cite{CoVaR2016} and \cite{girardi2013systemic} utilize GARCH models to capture the dynamic evolution of systemic risk contributions. \cite{white2015var} employ a combination of quantile regressions and GARCH to estimate CoVaR, and their method can be extended to a multivariate-quantiles condition. 

Alternatively, one may adopt distributional assumptions and utilize maximum likelihood techniques for CoVaR estimation. For instance, \cite{cao2013multi} employs a multivariate Student-$t$ distribution to estimate CoVaRs across firms. Similarly, \cite{bernardi2013multivariate} estimate CoVaR by employing a multivariate Markov switching model with a Student-$t$ distribution, effectively capturing heavy tails and nonlinear dependence. Furthermore, their methods can be extended to estimate CoVaR with a multivariate-quantiles condition. 

Copula models have gained popularity in CoVaR estimation due to their convenient ability to model the dependence between portfolio losses. For example, \cite{mainik2014dependence} present analytical results regarding CoVaR using copulas. \cite{karimalis2018measuring} provide a straightforward closed-form expression of the CoVaR for a wide range of copula families, allowing for the incorporation of time-varying exposures. \cite{oh2018time} introduce a new class of copula-based dynamic models for high-dimensional conditional distributions that facilitates multivariate-quantiles conditions. 

To accommodate nonlinear relationship between institutions' losses, \cite{hardle2016tenet} propose a nonlinear model where an institution's loss is a function with nonlinear properties, dependent on the aggregation of other institutions' losses. Based on this framework, they develop a CoVaR estimation method. \cite{cai2020functional} propose an alternative dynamic CoVaR estimation approach based on a functional coefficient vector autoregressive model. 

\begin{figure}[ht]
	\captionsetup{labelfont=bf}
	\caption{The Logic of Model-Based Methods \& Monte-Carlo Simulation Based Methods}
	\centering    
	\begin{tikzpicture}[node distance=6.5cm, every node/.style={minimum width=2.5cm, minimum height=1cm},myarrow/.style={-{Triangle[length=3mm,width=2mm]}}]
		\node (data) [,draw,rectangle,line width=1pt]{Real Data};
		\node (model) [right of=data,draw,rectangle,line width=1pt] {Model};
		\node (covar) [right of=model,draw,rectangle,line width=1pt] {CoVaR};
		\node (simulationdata) [below left of=covar, xshift=1.3cm, yshift=1.5cm,draw,rectangle,line width=1pt] {Simulation Data};
		
		\draw[->, shorten >=2pt, shorten <=2pt, line width=1.5pt, solid, draw opacity=0] (data) -- (model) node[midway, above] {Step (i): Estimation};
		\draw[-to, shorten >=2pt, shorten <=2pt, line width=1.5pt, dashed] let \p1 = ($(model)-(data)$) in ($(data)+(1.27cm,0.1cm)$) -- ($(model)+(-1.27cm,0.1cm)$);
		\draw[myarrow, shorten >=2pt, shorten <=2pt, line width=1.5pt, solid] let \p1 = ($(model)-(data)$) in ($(data)+(1.27cm,-0.1cm)$) -- ($(model)+(-1.27cm,-0.1cm)$);
		\draw[-to, shorten >=2pt, shorten <=2pt, line width=1.5pt, dashed] (model) -- (covar) node[midway, above]  {Step (ii): Calculation};
		\draw[myarrow, shorten >=2pt, shorten <=2pt, line width=1.5pt, solid] (model) -- (simulationdata) node[midway, left]  {Step (iii): Simulation~~};
		\draw[myarrow, shorten >=2pt, shorten <=2pt, line width=1.5pt, solid] (simulationdata) -- (covar) node[midway, right]  {~~Step (iv): Estimation};
	\end{tikzpicture}
	\caption*{\footnotesize \it Note: Model-based methods (indicated by the dashed line procedure) comprise Step~(i) followed by Step~(ii), whereas Monte-Carlo simulation based methods (indicated by the solid line procedure) involve Step~(i) followed by Step~(iii) and Step~(iv). Step~(ii) is often challenging to solve mathematically, requiring the model-based methods to simplify the models. Consequently, these simplifications result in biased CoVaR estimates. In contrast, simulation-based methods bypass Step~(ii) and offer greater flexibility in modeling.}
	\label{fig0}
\end{figure}

We refer to the above estimation approaches as \textit{model-based methods}, the logic of which is illustrated by Step~(i) and Step~(ii) shown in Figure~\ref{fig0}. Specifically, these methods involve fitting a predetermined statistical model to real financial data (i.e., Step~(i)) and mathematically calculating CoVaR based on the model (i.e., Step~(ii)). The advantage of model-based methods lies in their efficiency when the model is appropriately specified. However, if the model is misspecified, bias may be introduced that is difficult to remove. In order to obtain an analytical solution for CoVaR in Step~(ii), the aforementioned model-based methods simplify their models, which result in the inability to capture certain conventional financial models. For example, these models cannot accommodate the basic delta-gamma approximation model for complicated portfolios \citep{glasserman2004monte}, nor some advanced financial models such as the constant elasticity of variance model \citep{cox1976valuation,cox1996constant,schroder1989computing,black1976pricing,beckers1980constant}, the stochastic return model \citep{kim1996dynamic,wachter2002portfolio,merton1975optimum}, and the stochastic volatility model \citep{chacko2005dynamic,heston1993closed,campbell1999consumption}. Consequently, using these model-based methods may lead to potentially significant biases in CoVaR estimation.

Another approach involves utilizing \textit{Monte Carlo simulation-based methods} for CoVaR estimation, as first introduced by \cite{huanghong23}. The logic of these methods is illustrated by Steps (i), (iii), and (iv) in Figure \ref{fig0}. Specifically, the simulation-based methods also entail fitting a predetermined statistical model to real financial data (i.e., Step (i)), but they bypass Step (ii) by proceeding directly to Steps (iii) and (iv). The advantage of Monte Carlo simulation-based methods lies in their model flexibility, as they can accommodate complicated models capable of generating ample observations. Based on the simulation data, estimation methods can be established to obtain CoVaR. Furthermore, \cite{huanghong23} propose two estimation methods for Step (iv): batching estimation (BE) and importance sampling-inspired estimation (ISE). BE is also a two-step nonparametric method, but differs in that both steps involve order statistics, whereas our two-step method employs order statistics followed by kernel estimation. ISE is specific to delta-gamma approximation models and exhibits higher efficiency than BE with a faster rate of convergence. However, while the simulation-based methods demonstrate their efficiency in estimating CoVaR with a univariate-quantile condition, they are not applicable to CoVaR with a multivariate-quantiles condition.

Table \ref{Table0} summarizes the existing estimation methods. The current model-based methods lack model flexibility as they require simplifying their models for CoVaR calculation, although some of them can be utilized for multivariate-quantile conditions. On the other hand, the simulation-based method BE offers model flexibility but exhibits a slower rate of convergence $n^{-1/3}$ and is not applicable to multivariate-quantile conditions. The simulation-based method ISE, which has a faster rate of convergence compared to BE and is based on more advanced models than the model-based methods, is also not applicable to multivariate-quantile conditions. In contrast, the new two-step kernel nonparametric estimator in this paper provides model flexibility and is applicable to multivariate-quantile conditions. Its rate of convergence is $(n\Delta_n)^{-1/2}$, which is faster than that of BE for univariate-quantile conditions.

\begin{table}[ht]
	\centering
	\captionsetup{labelfont=bf}
	\renewcommand{\arraystretch}{1.5} 
	\begin{threeparttable}
		\caption{Comparison between Different Estimation Methods}
		\begin{tabular}{@{}ccccc@{}}
			\toprule
			\multicolumn{2}{c}{\makecell{Estimator}} & \multicolumn{1}{c}{\makecell{Model \\ Flexibility}} & \multicolumn{1}{c}{\makecell{Rate of \\ Convergence}} & \multicolumn{1}{c}{\makecell{Multivariate \\ Condition}} \\
			\midrule
			~Model-Based Methods & & no & - & -  \\
			~BE \citep{huanghong23} & & yes & $n^{-1/3}$ & no \\
			~ISE \citep{huanghong23} & & no & $n^{-1/2}$ & no \\
			~Two-Step Estimator ({\color{blue}This work}) & & yes & ~~~$(n\Delta_n)^{-1/2}$~~~ & yes \\
			\bottomrule
			\label{Table0}
		\end{tabular}
	\end{threeparttable}
\end{table}

The remainder of this paper is structured as follows. In Section~\ref{prob}, we rigorously formulate the estimation problem and present the assumptions that will be employed throughout the paper. In Section~\ref{kernel}, we introduce the two-step nonparametric estimator, establish its consistency and asymptotic normality, and conclude with a heuristic method for bandwidth selection. In Section~\ref{App}, we provide two numerical examples to substantiate our theoretical findings. Finally, in Section~\ref{conclusion}, we summarize the paper's main contributions. The detailed proofs are included in Appendix. 

\section{Problem Formulation}\label{prob}

In this section, we formally formulate the estimation of $\cv$ and introduce some notations that will be utilized throughout the paper. Let $X =(X_1,\ldots,X_m)^\top$ and $Y$ represent $m+1$ continuous random variables, and denote the joint density function of $(X,Y)$ as $f(x,y)$. For $x =(x_1,\ldots,x_m)^\top$, we define the marginal density function and cumulative distribution function of $X$ as follows:
$$
f_{  X } (  x )=\int_{\R} f(  x ,y)\dd y \quad \text{and} \quad F_{  X } (  x ) = \int_{-\infty}^{x_1} \cdots \int_{-\infty}^{x_m} f_{  X }(u_1,\ldots, u_m) \dd u_1 \cdots \dd u_m.
$$ 
These functions capture the individual characteristics of $X$ without considering $Y$. In the field of financial risk management, it is common practice to represent the losses of financial portfolios using continuous random variables \citep{hull2012risk}. For instance, $X_1,\ldots,X_m,Y$ may correspond to the losses incurred by different financial institutions. 

Based on Section 4.1 of \cite{durrett2019probability} or Section 3.3 of \cite{Ross}, when $f_{  X }(  x )>0$, we let
\begin{equation}\label{defofCP}
	F_{Y|  X }(y|  x ) \ =\ \pr \big\{ Y\leq y \big|   X =   x  \big\}  
	\ =\  \lim\limits_{\varepsilon\to 0}\frac{\pr\big\{Y\leq y, \|   X  -   x  \| \leq \varepsilon \big\}}{\pr\big\{ \|  X -   x  \| \leq \varepsilon \big\}} \ =\ \int_{-\infty}^{y} \frac{f(  x , v)}{f_{  X }(  x )} \mathrm{d}v 
\end{equation}
be the conditional distribution function, where $\|\cdot\|$ is the Euclidean metric. 
To simplify the notation, we let 
$$
g(  x ,y)=\int_{-\infty}^y f(  x ,v)\dd v,
$$ 
so we have $F_{Y|  X }(y|  x ) = g(  x , y)/ f_{  X }(  x )$. 

To define derivatives of multivariate function, we introduce the multi-index notation as follows. 
For $\ell=(\ell_1,\ldots,\ell_m)^\top \in\N^m$ (where $\N=\{0,1,2,\ldots \}$), $x=(x_1,\ldots,x_m)^\top \in\R^m$, we denote 
$$
|\ell|= \ell_1+\ell_2+\cdots+\ell_m, \quad
\ell != \ell_1 ! \ell_2 ! \cdots \ell_m !, \quad \text{and} \quad x^\ell=x_1^{\ell_1} x_2^{\ell_2}\cdots x_m^{\ell_m}.
$$ 
Notice that, we denote $x^{2\ell} = x_1^{2\ell_1} x_2^{2\ell_2}\cdots x_m^{2\ell_m}$. 
Then, we denote 
$$
\D^{\ell} f_X( x ) = \frac{\partial^{|\ell|} f_X}{\partial x_1^{\ell_1}\cdots \partial x_m^{\ell_m}}(  x ) \quad \text{and} \quad \D^{\ell} g(  x ,y) = \frac{\partial^{|\ell|} g}{\partial x_1^{\ell_1}\cdots \partial x_m^{\ell_m}}(  x , y)
$$ 
be the $\ell$-th partial derivative of $f_X$ and $g$ with respect to $  x $ respectively. These notations are common used in calculus for multivariate functions  and will be used throughout this paper. 

The most commonly used risk measure employed by financial institutions is value-at-risk (VaR), which is defined as a quantile. Specifically, the $\beta$-VaR (or $\beta$-quantile) of $Y$, denoted as $q_\beta(Y)$ with $\beta\in(0,1)$, satisfies
$$
\pr\Big\{ Y\leq q_\beta(Y) \Big\} = \beta.
$$
This expression implies that there is a $100 \times \beta\%$ confidence level associated with the statement that the loss of $Y$ does not exceed $q_\beta(Y)$. The concept of VaR was originally introduced by J.P. Morgan in the early 1990s and has since become a widely adopted risk measure in the global financial industry \citep{Jorion,DuffiePan,hull2012risk}. However, VaR, as a measure of risk for individual institutions, may not fully capture systemic risk within financial networks. This limitation arises from its inability to quantify the interconnectedness between nodes within the networks. 

In response to the financial crisis, \cite{CoVaR2016} introduced $\cv$ as a metric for measuring financial systemic risk, which satisfies
\begin{equation}
	\pr\left\{ Y\leq \cv_{\alpha_1,\beta} \Big| X_1= q_{1,\alpha_1}\right\} = \beta, 
\label{defCoVaR}
\end{equation}
where $\alpha_1,\beta\in(0,1)$. CoVaR represents a conditional quantile, specifically the $\beta$-quantile of $Y$ when $X_1$ reaches its $\alpha_1$-quantile ($X_1=q_{1,\alpha_1}$). It is important to note that $\cv_{\alpha_1,\beta}={\rm VaR}_\beta(Y)$ when the two random variables $X_1$ and $Y$ are independent. However, in situations where financial portfolio losses exhibit positive dependence, $\cv_{\alpha_1,\beta}$ is typically significantly larger than ${\rm VaR}_\beta(Y)$, indicating that tail risk during financial distress is considerably higher than during normal times. To quantify the sensitivity of the value-at-risk of $Y$ to the distress caused by the risk of $X_1$, they define 
$$
\Delta\cv_{\alpha_1,\beta} = \cv_{\alpha_1,\beta} -\cv_{0.5,\beta} \quad \mbox{for $\alpha_1>0.5$},
$$
which represents the increase in the value-at-risk of $Y$ when $X_1$ transitions from a median state $\{X_1=q_{1,0.5}\}$ to a crisis state $\{X_1=q_{1,\alpha_1}\}$. It is worth noting that when defining the crisis state, it is common to consider $\alpha_1$ close to one. Furthermore, empirical evidence presented by \cite{CoVaR2016} demonstrates that $\Delta\cv_{\alpha_1,\beta}$ effectively captures systemic risk and predicts the 2007-2009 financial crisis.

However, the above definition of CoVaR solely assesses node-to-node systemic risk, which measures the tail dependence between pairs of random variables $(X_1,Y)$. Nevertheless, during a financial crisis, it is not uncommon to observe simultaneous crises occurring across numerous financial institutions. Consequently, it becomes necessary to consider CoVaR that accounts for the simultaneous influence of risks from multiple correlated nodes. Therefore, it is natural to extend Equation \eqref{defCoVaR} to incorporate a multivariate-quantile condition, yielding the following formulation:
\begin{equation}\label{defCoVaRGeneral}
	\pr\Big\{ Y\leq \cvv \Big| X_1 = q_{1,\alpha_1},\ldots, X_m=q_{m,\alpha_m} \Big\} = \beta, 
\end{equation}
where $ {\alpha}=(\alpha_1,\ldots,\alpha_m)^\top$ and $q_{i,\alpha_i}$ is the $\alpha_i$-quantile of $X_i$ satisfying $\pr\left\{ X_i \leq q_{i,\alpha_i} \right\} = \alpha_i$, for $i=1,\ldots,m$. Furthermore, we define
$$
\Delta\cvv = \cvv -\cv_{(0.5,\ldots,0.5),\beta} \quad \mbox{for $\alpha_i\geq 0.5,~i=1,\ldots,m$}. 
$$
$\Delta\cvv$ effectively measures the increase in the value-at-risk of $Y$ when some or all of the variables $X_1,\ldots,X_m$ transition simultaneously from their median state to crisis states. 

Estimating CoVaR as defined by Equation \eqref{defCoVaRGeneral} poses challenges due to the unobservability of the event $\{ X_1 = q_{1,\alpha_1},\ldots, X_m=q_{m,\alpha_m} \}$. This unobservability has two implications: Firstly, since $X_1,\ldots,X_m$ are continuous random variables, the event has a probability of zero. Secondly, the quantiles $q_{1,\alpha_1},\ldots,q_{m,\alpha_m}$ are unobservable. To overcome the challenges, this paper proposes a two-step nonparametric method based on Monte-Carlo simulation. It is worth emphasizing that Monte-Carlo simulation methods have been extensively employed in the field of risk management \citep{glasserman2004monte,hull2012risk}. Additionally, \cite{huanghong23} propose simulation-based methods for estimating CoVaR under a univariate-quantile condition ($m=1$). However, there exists a significant gap in extending their methods to the multivariate-quantile condition. 

Suppose that we have observed an independent and identically distributed (i.i.d.) simulation sample from a validated model, denoted by
$$
\Big\{ \big(X_{1,1},\ldots,X_{1,m},Y_1\big),~\big(X_{2,1}, \ldots, X_{2,m},Y_2\big),~\ldots,~\big(X_{n,1},\ldots,X_{n,m},Y_n\big) \Big\}.
$$ 
The objective of this paper is to develop an estimator for $\cv$ based on this sample. In other words, we aim to construct an estimator for Step(iv) depicted in Figure~\ref{fig0}. Furthermore, given the typically large sample sizes available in Monte Carlo studies, we also investigate the asymptotic properties of the estimator as the sample size $n$ approaches infinity. 

To facilitate the development and the analysis of the $\cv$ estimator, we impose the following assumptions. Let $  q_\alpha =(q_{1,\alpha_1},\ldots,q_{m,\alpha_m})^\top$ where $q_{j,\alpha_j}$ is the $\alpha_j$-quantile of $X_j$, $j=1,\ldots,m$. 

\begin{assumption}\label{assu:dist}
	Let $X_1,X_2,\ldots,X_m,Y$ be continuous random variables. 
	Let $\alpha_1,\ldots,\alpha_m,\beta\in(0,1)$ be given, $\mathcal{X}\subset \R^m $ be a convex neighborhood of $  q_\alpha $, and $\mathcal{Y}\subset \R $ be a neighborhood of $\cvv$. 
	Then, 
	\begin{enumerate}[label={(\roman*)}]
	\item both $f_{X}(x)$ and $f(x ,y)$ are positive and continuous in $\mathcal{X}$ and $\mathcal{X}\times\mathcal{Y}$ respectively; 
	\item for any $y\in\mathcal{Y}$ and $\ell\in\N^m$ such that $|\ell|\leq1$, we have $\D^\ell f_{X}(  x )$ and $\D^\ell g(x,y)$ are continuous in~$\mathcal{X}$;
	\item for $  x \in\R^m$ and $\ell\in\N^m$ such that $|\ell|\leq 1$, we have $\D^\ell g(  x ,y)$ is uniformly continuous in $y\in\mathcal{Y}$;
	\item for $y\in\mathcal{Y}$ and $\ell\in\N^m$ such that $|\ell|\leq 1$, we have $\int_{\R^m} |\D^\ell g(  x ,y)|\dd    x  <\infty$.
	\end{enumerate}
\end{assumption}

Assumption \ref{assu:dist} is commonly found in the nonparametric statistics literature, as evident in works such as \cite{Pagan1999nonparametric,hardle2004nonparametric,li2007nonparametric}. Moreover, it is also prevalent in the financial engineering literature when analyzing the properties of VaR and $\cv$, as illustrated in papers by \cite{Hong09,huanghong23,LiuandHong}.  We can verify that Assumption \ref{assu:dist} holds for commonly used distribution functions. Under Assumption \ref{assu:dist}, it is apparent that $q_{i,\alpha_i}$ represents the unique value satisfying $\mathbb{P}\{X_i \leq q_{i,\alpha_i} \}=\alpha_i$ and $q_{i,\alpha_i} = F_{X_i}^{-1}(\alpha_i)$ for $i=1,\ldots,m$. Furthermore, Assumption \ref{assu:dist} guarantees that $\cv$ is the unique solution of Equation \eqref{defCoVaRGeneral}. In fact, for $x \in \mathcal{X}$, $F_{Y|X}(y|x)$ is a differentiable function of $y$, and the conditional density $f_{Y|X}(y|x)$ satisfies
\begin{equation*}\label{Fpiany}
f_{Y|  X }(y|  x ) = \frac{\partial}{\partial y} F_{Y|  X }(y|  x ) = \frac{f(  x ,y)}{f_{  X }(  x )} > 0 \quad \text{when } y\in\mathcal{Y}. 
\end{equation*}
Consequently, for any $x \in \mathcal{X}$, we have an inverse function $y=F_{Y|X}^{-1}(\beta|x)$, and
\begin{equation}\label{ConQuan}
\cvv = F^{-1}_{Y|  X }(\beta\,|\,  q_\alpha ). 
\end{equation}

\begin{assumption}\label{assu:kernel1}
The kernel function $K$ is a bounded, symmetric and compactly supported probability density function that satisfying 
	~(i) $tK(t)\to 0$ as $|t|\to\infty$; 
	~(ii) $\int_{\R} K(t) \dd t=1$; 
	~(iii) $\int_{\R} |K(t)| \dd t<\infty$; 
	~(iv) $\int_{\R} t K(t) \dd t=0$; 
	~(v) $\int_{\R} t^2K(t) \dd t <\infty$; 
    ~(vi) for any $\ell\in\N^m$ such that $|\ell|\leq 1$, we have $[\D^\ell \prod_{j=1}^m K(x_j)]^2$ is bounded and integrable, and $\lim_{\|x\| \to\infty} \|x\| \cdot \D^\ell \prod_{j=1}^m K(x_j) = 0$.
\end{assumption}

Assumption~\ref{assu:kernel1} is widely employed in nonparametric statistics as a standard assumption for the kernel function, as supported by \cite{Pagan1999nonparametric,hardle2004nonparametric,li2007nonparametric}. Note that, a kernel satisfying conditions (i)--(v) of Assumption~\ref{assu:kernel1} corresponds to a second-order kernel, while the condition (vi) of Assumption~\ref{assu:kernel1} requires that $\prod_{j=1}^m K(x_j)$ forms a Parzen-Rosenblatt kernel \citep{abdous1998pointwise}. For the sake of simplicity in notation, we define the product kernel function as
\begin{equation}\label{Notation1}
W_{n,i}(  x ) \ = \ \prod_{j=1}^m  K\left( \frac{x_j-X_{i,j}}{\delta_{n,j}} \right) \quad \text{and} \quad \D^\ell W_{n,i}(  x ) \ = \ \frac{\partial^{|\ell|} W_{n,i}}{\partial x_1^{\ell_1}\cdots \partial x_m^{\ell_m}}(  x ), 
\quad i=1,\ldots,n, 
\end{equation}
where $|\ell|=\sum_{i=1}^m \ell_i$ and $\delta_{n,j}$ for $j=1,\ldots,m$ are the bandwidth parameters. Denote 
$$
\Delta_n=\prod_{j=1}^m \delta_{n,j} \quad \text{and} \quad \delta_n= \big(\delta_{n,1},\delta_{n,2},\ldots,\delta_{n,m} \big)^\top. 
$$ 
These product kernels are commonly used in the estimation of multivariate density in the nonparametric statistics literature. 

\begin{assumption}\label{assu:kernel3}
As $n\to\infty$, we have $\delta_{n,j}\to 0 $ for all $j=1,\ldots,m$, and $n\Delta_{n}\to\infty$. 
\end{assumption}

Assumption \ref{assu:kernel3} is a typical assumption for the bandwidth parameters in nonparametric statistics \citep{Pagan1999nonparametric,hardle2004nonparametric,li2007nonparametric}. This assumption requires that all $m$ bandwidths approach zero as the sample size $n\to\infty$, ensuring that the bias of the kernel estimation diminishes as $n\to\infty$. Additionally, the assumption $n\Delta_{n}\to\infty$ imposes a condition that the bandwidths cannot decrease too rapidly, guaranteeing the convergence of the kernel estimation with a rate of $(n\Delta_n)^{-1/2}$. 

In this paper, we employ the notation $Y_n=\mathcal{O}_p(X_n)$ to signify that, for any $\varepsilon>0$, there exists a constant $c>0$ such that $\pr\{|Y_n/X_n|>c\} \leq \varepsilon$ for all $n\in\mathbb{N}$. We utilize the notation $X_n\Rightarrow X$ to indicate that $X_n$ converges in distribution to $X$. Furthermore, we utilize $\ceil{\cdot}$ to denote the ceiling function, where $\ceil{x}$ represents the smallest integer greater than or equal to $x$ for $x\in\R $; we utilize $\mathbb{I}\{ \cdot\}$ to denote the indicator function. 

\section{A Two-Step Nonparametric Estimation}\label{kernel}

To estimate $\cv$, we examine the conditional distribution function $F_{Y|  X }(y|  x ) = \pr\{Y \leq y\, |\,   X =  x \}$ for $  x  \in \mathcal{S}$, where $\mathcal{S}=\{  x \in\R^m :f_{  X }(  x )>0\}$ represents the support of $  X $.  According to  Equation~\eqref{defofCP}, we have $\pr\{Y \leq y\, |\,   X =  x \} = g(  x , y)/ f_{  X }(  x )$. Drawing upon the classical theory of kernel density estimation, we could estimate $f_{  X }(  x )$ using ${1\over n\Delta_n} \sum_{i=1}^{n} W_{n,i}(  x )$, and estimate $ g(  x ,y)$ using $ {1\over n\Delta_n} \sum_{i=1}^{n} W_{n,i}(  x ) \mathbb{I}\{ Y_i \leq y\}$, where $W_{n,i}(x)$ is defined in Equation~\eqref{Notation1}. For further details, please refer to Section 6 of \cite{li2007nonparametric}. Consequently, the estimation of the conditional distribution function $F_{Y|X}(y|x)$ can be accomplished through the following sample distribution function
\begin{align}
	\hat F_n(y,  x ) \ &=\  \frac{ {1\over n\Delta_n} \sum_{i=1}^{n} W_{n,i}(  x ) \mathbb{I}\{Y_i\leq y\} }{ {1\over n\Delta_n} \sum_{i=1}^{n} W_{n,i}(  x )} \label{eqn:emp:kernel} \\ 
	&= \ \sum_{i=1}^n w_i(  x ) \mathbb{I}\{Y_i\leq y\}, \nonumber
\end{align}
where $w_i(  x ) = W_{n,i}(  x )/\sum_{i=1}^{n} W_{n,i}(  x )$. 

Based on Equation~\eqref{eqn:emp:kernel}, we propose a two-step approach to estimate $\cv$. Consider a dataset of $n$ independent and identically distributed (i.i.d.) samples generated through Monte Carlo simulation: $\{(X_{1,1},\ldots,X_{1,m},Y_1), (X_{2,1}, \ldots, X_{2,m},Y_2),\ldots, (X_{n,1},\ldots,X_{n,m},Y_n)\}$. 
\begin{description}	
	\item[Step 1.]  For each $j=1,2,\ldots,m$, we arrange $X_{1,j},X_{2,j},\ldots,X_{n,j}$ in ascending order as 
	$$
	X_{(1),j}\leq X_{(2),j}\leq \cdots \leq X_{(n),j},
	$$ 
	where $X_{(i),j}$ denotes the $i$-th smallest value, $i=1,2,\ldots,n$. We denote $\hat{q}_{j,\alpha_j} = X_{(\ceil{\alpha_j n}),j}$ for $j=1,2,\ldots,m$, and $ \hat{q}_\alpha  = (\hat{q}_{1,\alpha_1},\hat{q}_{2,\alpha_2},\ldots,\hat{q}_{m,\alpha_m})^\top$. 
	
	\item[Step 2.] Let $w_i = w_i(  \hat{q}_\alpha )$ for $i=1,2,\ldots,n$. Then, we obtain $n$ pairs $\{(Y_1,w_1), (Y_2,w_2), \ldots, (Y_n,w_n)\}$. We arrange $ Y_1,Y_2,\ldots, Y_n$ in ascending order as 
	$$ 
	Y_{[1]}\le Y_{[2]}\le\cdots\le Y_{[n]}, 
	$$ 
	where $Y_{[i]}$ denotes the $i$-th smallest value, $i=1,2,\ldots,n$. Here, the subscript $(i),j$ denotes the order of $\{X_{1,j},X_{2,j}\ldots,X_{n,j}\}$, and the subscript $[i]$ denotes the order of $\{Y_{1},Y_{2},\ldots,Y_{n}\}$. Furthermore, let $\pi_i =\sum_{j=1}^{i}  w_{[j]}$ for $i=1, \ldots, n$. Then, there exists $m \leq n$ such that $\pi_{m-1} \leq \beta$ and $\pi_{m}>\beta$. We define $$\hat{Y}= Y_{[m]}$$ which is the two-step nonparametric estimator of $\cvv$. 
\end{description}

According to Equation~\eqref{ConQuan}, $\cvv$ is the inverse of the distribution function $F_{Y|X}(y, x)$ at $x=q_\alpha$ and $y=\beta$, whereas the two-step estimator $\hat{Y}$ is the inverse of the sample distribution function $\hat F_n(y, x)$ at $x=\hat{q}_\alpha$ and $y=\beta$. Step~1 involves the utilization of order statistics $\hat{q}_\alpha$ to estimate $q_\alpha$, a topic has been extensively studied in the classical theory of order statistics \citep{Serfling1980}. On the other hand, Step~2 employs the kernel method to estimate the conditional quantile $\cvv$. It is important to note that the kernel method for estimating a conditional quantile with a fixed and known value condition, i.e., $\{X=x\}$ where $x\in\R^m$ is fixed and known, has been thoroughly studied in the classical theory of nonparametric statistics \citep{li2007nonparametric}. Our two-step approach integrates these two research areas. However, in Step 2, the conditional quantile $\cvv$ deviates from the conventional definition, as $q_\alpha$ in the condition $\{x=q_\alpha\}$ is unobservable (unknown). This distinction introduces new challenges to the statistical analysis. 

A major challenge in estimating $\cvv$ is that the quantile $q_\alpha$ cannot be directly observed. To address this issue, we approximate $q_\alpha$ using the sample quantile $\hat{q}_\alpha$, which is a strongly consistent estimator of $q_\alpha$ as $n\to\infty$, as discussed in Sections 2.3 and 2.4 of \cite{Serfling1980}. Subsequently, we incorporate this approximation into the kernel estimation of the conditional quantile. However, when analyzing the asymptotic properties of the two-step approach, it becomes necessary to simultaneously evaluate the errors arising from quantile estimation (Step 1) and the estimation of the conditional quantile (Step 2).

It is also essential to emphasize the rationale behind employing the kernel method to estimate $\cvv$. As highlighted by \cite{huanghong23}, a significant challenge in estimating $\cvv$ is the unobservability of the event $\{X=q_\alpha\}$, which has a probability of zero. A natural approach is to utilize the data within a small $\varepsilon$-neighborhood of $q_\alpha$, i.e., $\{q_\alpha -\varepsilon \leq X \leq q_\alpha +\varepsilon\}$, and then sending $\varepsilon$ to zero. The kernel method serves as one such technique.  It assigns weights, denoted as $w_i$ for $i=1,\ldots,n$, to each data point of $Y$, such that a higher (resp. lower) weight is assigned when the corresponding data point of $X$ is closer to (resp. farther from) $\hat{q}_\alpha$. By employing these weighted data points of $Y$ and sending the bandwidths to zero, we estimate the conditional distribution function $F_{Y|X}(y|x)$ and subsequently estimate $\cvv$ by taking its inverse.

In the remainder of this section, we analyze the asymptotic properties of the two-step estimator as the sample size $n$ tends to infinity and the bandwidths $\delta_{n,j}$, $j=1,\ldots,m$, tend to zero. In Section~\ref{consistency}, we prove the consistency of the two-step estimator. In Section~\ref{asymptotic_normality}, we prove the asymptotic normality of the two-step estimator and then provide guidelines for selecting the appropriate bandwidths.

\subsection{Consistency}\label{consistency}

Now, we establish the consistency of the two-step estimator $\hat{Y}$. To analyze the estimator, we begin with the the conditional distribution function $F_{Y|  X }(y\,|\,  q_\alpha )= \pr\{ Y\leq y\,|\,  X =   q_\alpha  \}$ and its estimator $\hat F_{n}(y,  \hat{q}_\alpha )$ as given by Equation \eqref{eqn:emp:kernel}. To simplify Equation \eqref{eqn:emp:kernel}, we introduce the following definitions:  
$$
\bar{R}_n(  x ,y) = {1\over n\Delta_n} \sum_{i=1}^{n} W_{n,i}(  x ) \mathbb{I}\{Y_i\leq y\} \quad \text{and} \quad \bar{Q}_n(  x )= {1\over n\Delta_n} \sum_{i=1}^{n} W_{n,i}(  x ). 
$$
Consequently, we have $\hat F_n(y,  x )= \bar{R}_n(  x ,y)/\bar{Q}_n(  x )$. Therefore, to establish the convergence of $\hat {F}_{n}(y,  \hat{q}_\alpha )$ to $F_{Y|  X }(y|  q_\alpha )$ as $n\to\infty$, it is necessary to demonstrate the convergence of $\bar{R}_n(   \hat{q}_\alpha  ,y)$ to $g(  q_\alpha ,y) $ and $\bar{Q}_n(   \hat{q}_\alpha  )$ to $ f_{  X }(  q_\alpha ) $, respectively, as $n\to\infty$. We decompose the estimation errors of $\bar{R}_n(   \hat{q}_\alpha  ,y)$ and $\bar{Q}_n(   \hat{q}_\alpha  )$ into three parts respectively: 
\begin{align}
\bar{R}_n(   \hat{q}_\alpha  ,y) -  g(  q_\alpha ,y) 
	\ &=\ \underbrace{ \bar{R}_n(   \hat{q}_\alpha  ,y) -  \bar{R}_n(   q_\alpha  ,y) }_{\text{Error~I--1}} 
	 + \underbrace{ \bar{R}_n(   q_\alpha  ,y)  -  \E[\bar{R}_n(   q_\alpha  ,y)] }_{\text{Error~II--1}} + \underbrace{ \E[\bar{R}_n(   q_\alpha  ,y)]  -  g(  q_\alpha ,y) }_{\text{Error~III--1}}, \label{num} \\
\bar{Q}_n(   \hat{q}_\alpha  ) -  f_{  X }(  q_\alpha ) 
\ &=\ \underbrace{ \bar{Q}_n(   \hat{q}_\alpha ) - \bar{Q}_n(   q_\alpha )  }_{\text{Error~I--2}} 
 + \underbrace{ \bar{Q}_n(   q_\alpha )  -  \E[\bar{Q}_n(   q_\alpha )] }_{\text{Error~II--2}} + \underbrace{ \E[\bar{Q}_n(   q_\alpha )]  -  f_{  X }(  q_\alpha ) }_{\text{Error~III--2}}, \label{num01} 
\end{align}
Notice that, Errors I--1 and I--2 are caused by the error of the quantile estimator, i.e., $(\hat{q}_\alpha -q_\alpha)$, in Step~1, whereas Errors II--1 and II--2 are caused by the variations of the empirical distribution estimators $\bar{R}_n(   q_\alpha  ,y)$ and $\bar{Q}_n(   q_\alpha )$ respectively, and Errors III--1 and III--2 are caused by the bias of the empirical distribution estimators $\bar{R}_n(   q_\alpha  ,y)$ and $\bar{Q}_n(   q_\alpha )$ respectively. 

In the following lemma, we examine the mean and variance of $\bar{R}_n(  q_\alpha ,y)$ and $\bar{Q}_n(q_\alpha)$, which establishes the convergence of Errors III--1 and III--2 to zero as $n\to\infty$. Furthermore, it yields even stronger results than the convergence of Errors III--1 and III--2, which are essential for establishing the consistency and asymptotic normality of the latter. 
The detailed proof is provided in the appendix~\ref{Lemma1}.

\begin{lemma}\label{lemmaofMV}
Suppose that Assumptions \ref{assu:dist}--\ref{assu:kernel3} hold. Then, for any $\ell\in \N^m$ such that $|\ell|\leq 1$, we have 
\begin{align}
	\lim _{n \rightarrow \infty} \E \left[ \frac{1}{n \Delta_n  } \sum_{i=1}^n \D^\ell W_{n,i}(  q_\alpha ) \right]   \ &= \ \D^\ell f_X(  q_\alpha ), \label{lem1:eqn:10} \\
	\lim _{n \rightarrow \infty} n \delta_n^{2 \ell} \Delta_n \cdot  \Var \left(\frac{1}{n \Delta_n  } \sum_{i=1}^n \D^\ell W_{n,i}(  q_\alpha ) \right)  \ &= \ f_X(  q_\alpha ) \int_{\R^m} \Big[\D^\ell \prod_{j=1}^m  K ( x_j )  \Big]^2 \dd  x, \label{lem1:eqn:20}
\end{align}
and for any $\ell\in \N^m$ such that $|\ell|\leq 1$ and any sequence $y_n\in\mathcal{Y}$ such that $y_n\to y\in\mathcal{Y}$ as $n \to \infty$, we have 
\begin{align}
	 \lim _{n \rightarrow \infty} \E \left[ \frac{1}{n \Delta_n  } \sum_{i=1}^n \D^\ell W_{n,i}(  q_\alpha ) \cdot \mathbb{I}\{Y_i\leq y_n\} \right]   \ &= \ \D^\ell g(  q_\alpha ,y ), \label{lem1:eqn:1} \\
	 \lim _{n \rightarrow \infty} n \delta_n^{2 \ell} \Delta_n \cdot  \Var \left(\frac{1}{n \Delta_n  } \sum_{i=1}^n \D^\ell W_{n,i}(  q_\alpha ) \cdot \mathbb{I}\{Y_i\leq y_n\}  \right)  \ &= \ g(  q_\alpha , y) \int_{\R^m} \Big[\D^\ell \prod_{j=1}^m  K ( x_j )  \Big]^2 \dd  x. \label{lem1:eqn:2} 
\end{align}
Moreover, for any $\gamma>1$, we have 
\begin{equation}\label{lem1:eqn:3}
\lim _{n \rightarrow \infty} \E \left[\frac{1}{n\Delta_n}\sum_{i=1}^n [W_{n,i}(  q_\alpha )]^\gamma \cdot \mathbb{I}\{Y_i\leq y_n\} \right] \ =\ g(  q_\alpha ,y) \int_{\R^m} \Big[ \prod_{j=1}^m  K ( x_j )  \Big]^\gamma \dd x. 
\end{equation}
\end{lemma}

As demonstrated in Equations~\eqref{lem1:eqn:10} and \eqref{lem1:eqn:20}, in the case where $|\ell|=0$, the mean of $\bar{Q}_n(q_\alpha)$ converges to $f_X(q_\alpha)$, and the variance of $\bar{Q}_n(q_\alpha)$ is of the order $(n\Delta_n)^{-1}$. Similarly, as shown in Equations~\eqref{lem1:eqn:1} and~\eqref{lem1:eqn:2}, when $|\ell|=0$, the mean of $\bar{R}_n(q_\alpha,y)$ tends to $g(q_\alpha,y)$, and the variance of $\bar{R}_n(q_\alpha,y)$ is of the order $(n\Delta_n)^{-1}$. Additionally, Lemma~\ref{lemmaofMV} provides a more comprehensive result that encompasses cases where $|\ell|=1$ and any converging sequence $y_n\to y$. 
In the existing literature, the kernel estimation of $\D^\ell f_X(q_\alpha)$ often employs $\frac{1}{n \Delta_n } \sum_{i=1}^n \D^\ell W_{n,i}(q_\alpha)$, while the kernel estimation of $\D^\ell g(q_\alpha,y)$ often employs $\frac{1}{n \Delta_n } \sum_{i=1}^n \D^\ell W_{n,i}(q_\alpha) \mathbb{I}\{Y_i\leq y\}$, as discussed in \cite{Pagan1999nonparametric, bhattacharya1967estimation, schuster1969estimation}. Consequently, Lemma \ref{lemmaofMV} furnishes the means and variances of the kernel estimators for $\D^\ell f_X(q_\alpha)$ and $\D^\ell g(q_\alpha,y)$ respectively, aligning with this established intuition in the literature. 

As a direct consequence of Equations \eqref{lem1:eqn:10} and \eqref{lem1:eqn:1}, Errors III--1 and III--2 converge to zero as $n$ approaches infinity. This implies that $\bar{R}_n(q_\alpha,y)$ and $\bar{Q}_n(q_\alpha)$ serve as asymptotically unbiased estimators of $g(q_\alpha,y)$ and $f_X(q_\alpha)$, respectively. In the subsequent lemma, we establish that, utilizing Chebyshev's inequality (Theorem 1.6.4 of \citealt{durrett2019probability}) and Lemma~\ref{lemmaofMV}, Errors II--1 and II--2 tend to zero in probability as $n$ tends to infinity. The comprehensive proof can be found in Appendix~\ref{Lemma2}. 

\begin{lemma}\label{2error0}
	Suppose that Assumptions \ref{assu:dist}--\ref{assu:kernel3} hold. Then, for $y\in\mathcal{Y}$, we have 
	$\bar{R}_n(   q_\alpha  ,y)  -  \E[\bar{R}_n(   q_\alpha  ,y)]\to 0$ in probability and $ \bar{Q}_n(   q_\alpha )  -  \E[\bar{Q}_n(   q_\alpha )] \to 0$ in probability as $n \rightarrow\infty$. 
\end{lemma}

In next lemma, we prove that Errors I--1 and I--2 converge to zero in probability as $n\to\infty$. The more detailed proofs are included in the appendix~\ref{Lemma3}. 

\begin{lemma}\label{2error}
	Suppose that Assumptions \ref{assu:dist}--\ref{assu:kernel3} hold, and for all $\ell \in \N^m$ such that $|\ell|=1$ we have $n\delta_{n}^{2\ell} \Delta_n \to \infty$ as $n\to\infty$. 
	Then, for $y\in\mathcal{Y}$, we have $\bar{R}_n(   \hat{q}_\alpha  ,y) -  \bar{R}_n(   q_\alpha  ,y)\to 0$ in probability and $\bar{Q}_n(   \hat{q}_\alpha ) - \bar{Q}_n(   q_\alpha )\to 0$ in probability as $n \rightarrow\infty$. 
\end{lemma}

Lemma \ref{2error} shows that the error caused by the quantile estimation (i.e., $\hat{q}_\alpha-q_\alpha$ in Step 1) will vanishes as $n\to\infty$. The following remark discusses the assumption $n\delta_{n}^{2\ell} \Delta_n \to \infty$. 

\begin{remark}\label{remark1}
As demonstrated in Lemma~\ref{lemmaofMV}, the assumption $n\delta_{n}^{2\ell} \Delta_n \to \infty$ guarantees the variance of both $\frac{1}{n \Delta_n } \sum_{i=1}^n \D^\ell W_{n,i}(q_\alpha)$ and $\frac{1}{n \Delta_n } \sum_{i=1}^n \D^\ell W_{n,i}(q_\alpha) \mathbb{I}\{Y_i\leq y\}$ converging to zero. This requirement implies that the bandwidths $\delta_{n,j}$, $j=1,\ldots,m$, should not converge to zero too rapidly. For instance, when $\delta_{n,j}=b_n$ for $j=1,\ldots,m$, the assumption $n\delta_{n}^{2\ell} \Delta_n \to \infty$ as $n\to\infty$ can be expressed as $n b_{n}^{2+m} \to \infty$ as $n\to\infty$. This implies that if $b_n$ converges to zero at a slower rate than $n^{-1/(2+m)}$, both Errors II--1 and II--2 converge to zero in probability as $n\to\infty$.
\end{remark}

By combining the aforementioned lemmas and utilizing Slutsky's theorem (Section 1.5.4 of \citealt{Serfling1980}), we can directly establish the following theorem regarding the consistency of $\hat{F}_{n}(y, \hat{q}_\alpha)$. 

\begin{theorem}\label{consistencyofFbar}
	Suppose that Assumptions \ref{assu:dist}--\ref{assu:kernel3} hold, and for all $\ell \in\N^m$ such that $|\ell|=1$ we have $n\delta_{n}^{2\ell} \Delta_n \to \infty$ as $n\to\infty$. 
	Then, for $y\in\mathcal{Y}$, we have $\hat{F}_{n}(y,  \hat{q}_\alpha ) \rightarrow F_{Y|X}(y|  q_\alpha )$ in probability as $n\rightarrow\infty$.
\end{theorem}

Theorem \ref{consistencyofFbar} establishes the consistency of $\hat{F}_{n}(y, \hat{q}_\alpha)$ as an estimator of $F_{Y|X}(y|q_\alpha)$. Based on this theorem, we can demonstrate the consistency of the two-step estimator $\hat{Y}$. In fact, it is evident that $\hat{Y}$ satisfies the following equation:
\begin{equation*}
	\hat{Y}= \inf\{y\in\R : \hat{F}_{n}(y,  \hat{q}_\alpha )\geq \beta\}.
\end{equation*}
Let $\hat{F}^{-1}_{n}(z, \hat{q}_\alpha) = \inf\{y\in\R : \hat{F}_{n}(y, \hat{q}_\alpha) \geq z\}$. Consequently, we have $\hat{Y} = \hat{F}^{-1}_{n}(\beta, \hat{q}_\alpha)$. As demonstrated in Equation~\eqref{ConQuan}, we have $\cvv = F^{-1}_{Y|X}(\beta|q_\alpha)$. Hence, by taking the inverse of both $\hat{F}_{n}(y, \hat{q}_\alpha)$ and $F_{Y|X}(y|q_\alpha)$ in Theorem \ref{consistencyofFbar}, we can establish the convergence of $\hat{Y}$ to $\cvv$. The detailed proof is provided in Appendix~\ref{Theorem2}.

\begin{theorem}\label{thm34}
	Suppose that Assumptions \ref{assu:dist}--\ref{assu:kernel3} hold, and for all $\ell \in\N^m$ such that $|\ell|=1$ we have $n\delta_{n}^{2\ell} \Delta_n \to \infty$ as $n\to\infty$. 
	Then, we have $\hat{Y} \rightarrow \cvv$ in probability as $n\rightarrow\infty$. 
\end{theorem}

As pointed out in the Introduction, \cite{huanghong23} employ a batching estimator to estimate $\cv$ with a univariate-quantile condition. Their approach only utilizes the data conditional on $\{X=  \hat{q}_\alpha \}$, whereas our kernel method incorporates the entire dataset and assigns higher weights to data points in close proximity to $  \hat{q}_\alpha $. The consistency established in Theorem \ref{thm34} does not demonstrate the advantages of our two-step estimator nor provide guidelines for selecting the bandwidths $\delta_{n,j}$, $j=1,\ldots,m$. To address these issues, we need to analyze the rate of convergence of the two-step estimator and investigate its asymptotic distribution. 

\subsection{Asymptotic Normality}\label{asymptotic_normality}

To further examine the asymptotic normality of the two-step estimator, it is necessary to appropriately scale (normalize) the error term $(\hat{Y}-\cvv)$. Drawing inspiration from the conventional kernel estimation of conditional quantiles and empirical distributions, the scaling parameter should be chosen as $(n\Delta_n)^{1/2}$. In the subsequent analysis, we scale the errors of both $(\hat{Y}-\cvv)$ and $(\hat{F}_n-\beta)$ by $(n\Delta_n)^{1/2}$ and establish the asymptotic normality, thereby validating the above heuristic argument. 

To establish the asymptotic normality of $(n\Delta_n)^{1/2}(\hat{Y}-\cvv)$, we employ a transformation that allows us to analyze the empirical distribution $\hat{F}_n$. Let $y_n=\cvv + z \sigma_Y(n \Delta_n)^{-1/2}$, for $z \in\R $ and $\sigma_Y>0$, such that $y_n\to \cvv$ as $n\to\infty$ by Assumption~\ref{assu:kernel3}. Consequently, we obtain
\begin{eqnarray}
\lefteqn{ \pr \left\{ \frac{\sqrt{n\Delta_n}}{\sigma_Y} \big( \hat{Y}-\cvv \big) \leq z \right\} } \label{CLT:results} \\
&=&  \pr\Big\{ \hat{Y} \leq y_n \Big\} \ = \ \pr\Big\{ \hat{F}_n(y_n,   \hat{q}_\alpha )\geq \beta \Big\} \ =\ \pr \left\{ \sqrt{n\Delta_n} \Big[\hat{F}_n(y_n,   \hat{q}_\alpha )- \beta  \Big]\geq 0 \right\}, ~~~~ \label{CLT:eqn10}
\end{eqnarray}
where $\sigma_Y$ is the standard deviation of $\hat{Y}$. Our objective is to prove that the final term in Equation~\eqref{CLT:eqn10} towards the distribution function of a normally distributed random variable as $n$ approaches infinity, which enables us to infer that Equation~\eqref{CLT:results} follows an asymptotic normal distribution as $n$ tends to infinity.

Recalling that $\hat F_n(y_n,  \hat{q}_\alpha )= \bar{R}_n(  \hat{q}_\alpha ,y_n)/\bar{Q}_n(  \hat{q}_\alpha )$ and $\beta=F_{Y|X}(\cv,  q_\alpha )=g(  q_\alpha ,\cv)/f_X(  q_\alpha )$, we can express the following: 
\begin{eqnarray}
	\lefteqn{ \hat F_n(y_n,  \hat{q}_\alpha ) - \beta} \nonumber \\
	&=& \frac{1}{\bar{Q}_n(  \hat{q}_\alpha) f_X(  q_\alpha )}  \Big[ f_X (  q_\alpha )\left(\bar{R}_n(  \hat{q}_\alpha , y_n)-g(  q_\alpha , \cv)\right)-g(  q_\alpha , \cv)\left(\bar{Q}_n(  \hat{q}_\alpha )-f_X(  q_\alpha )\right) \Big] \nonumber \\
	&=& \frac{1}{\bar{Q}_n(  \hat{q}_\alpha) f_X(  q_\alpha )}  \Big[ f_X (  q_\alpha )\big( \mbox{Errors I--1 $+$ II--1 $+$ III--1}\big)-g(  q_\alpha , \cv)\big( \mbox{Errors I--2 $+$ II--2 $+$ III--2}\big) \Big], \nonumber
\end{eqnarray}
where Error I--1 $=\bar{R}_n(  \hat{q}_\alpha ,y_n) -\bar{R}_n(  q_\alpha ,y_n)$, Error II--1 $=\bar{R}_n(  q_\alpha ,y_n)- \E[\bar{R}_n(  q_\alpha ,y_n)] $, Error III--1 $=\E[ \bar{R}_n(  q_\alpha ,y_n) ] - g(  q_\alpha ,\cv)$, and Errors I--2, II--2, and III--2 are given by Equation~\eqref{num01}. It is important to note that Errors I--1, I--2, and I--3 in this subsection slightly differ from those in Equation~\eqref{num} due to the introduction of the sequence $y_n$ here. However, to avoid introducing excessive notation, we will continue to use the terms Error I--1, Error I--2, and Error I--3 in this subsection.
By substituting the above expression into Equation~\eqref{CLT:eqn10} and obtaining the asymptotic results for the following three components:
\begin{enumerate}[label={\it Component~(\Roman*)}, align=left, leftmargin=*]
	\item $\sqrt{n\Delta_n} \Big[ f_X (  q_\alpha )\big( \mbox{Error I--1}\big)-g(  q_\alpha , \cv)\big( \mbox{Error I--2}\big) \Big]$, 
	\item $\sqrt{n\Delta_n} \Big[ f_X (  q_\alpha )\big( \mbox{Error II--1}\big)-g(  q_\alpha , \cv)\big( \mbox{Error II--2}\big) \Big]$,  
	\item $\sqrt{n\Delta_n} \Big[ f_X (  q_\alpha )\big( \mbox{Error III--1}\big)-g(  q_\alpha , \cv)\big( \mbox{Error III--2}\big) \Big]$, 
\end{enumerate}
we can analyze the asymptotic result of the last term in Equation~\eqref{CLT:eqn10}. Subsequently, we develop the following three lemmas to examine Components~(I)--(III) individually. 

The following lemma establishes the convergence rates of Errors III--1 and III--2, and then obtains the asymptotic result of Component~(III). The detailed proof is provided in Appendix~\ref{Lemma4}.

\begin{lemma}\label{CRerrorIII}
	Suppose that Assumptions \ref{assu:dist}--\ref{assu:kernel3} hold. Then, we have 
	\begin{equation}\label{proofofCLT30}
		\E\Big[\bar{Q}_n(  q_\alpha ) \Big] - f_X(  q_\alpha ) \ = \ \frac12 \sum_{j=1}^m  \delta_{n,j}^2 \frac{\partial^2 f_X}{\partial x_j^2}(q_\alpha)   \int_{\R} t^2K( t)\dd t  + \sum_{j=1}^m o( \delta_{n,j}^2 ) \quad \mbox{as $n\to\infty$},
	\end{equation}
	and for any sequence $y_n\in\mathcal{Y}$ such that $y_n\to y\in\mathcal{Y}$ as $n \to \infty$, we have 
	\begin{eqnarray}
\lefteqn{ \E\Big[ \bar{R}_n(  q_\alpha ,y_n) \Big] - g(  q_\alpha ,y) }\nonumber \\
&=&  \frac12 \sum_{j=1}^m  \delta_{n,j}^2 \frac{\partial^2 g}{\partial x_j^2}(q_\alpha,y)   \int_{\R} t^2K( t)\dd t  +  f(  q_\alpha ,y)  (y_n-y) +  \sum_{j=1}^m o( \delta_{n,j}^2 ) + o(y_n-y) \quad \mbox{as $n\to\infty$}.  ~~~~~\label{proofofCLT3}
	\end{eqnarray}
\end{lemma}

Lemma \ref{CRerrorIII} demonstrates that the bias of $\bar{Q}_n(q_\alpha)$ can be approximated by the first term in the right-hand-side of Equation~\eqref{proofofCLT30}. To achieve a reduction in bias for the estimator $\bar{Q}_n(q_\alpha)$, we employ the bandwidth $\delta_{n,j}$ specific to each $X_j$, progressively reducing it towards zero. Consequently, we anticipate a corresponding reduction in bias on the order of $\delta_{n,j}^2$. Similarly, Lemma~\ref{CRerrorIII} demonstrates that the bias of $\bar{R}_n(q_\alpha,y_n)$ can be approximated by the first two terms in the right-hand-side of Equation~\eqref{proofofCLT3}, with an expected decrease in bias on the order of $\delta_{n,j}^2$ as the bandwidth $\delta_{n,j}$ approaches zero. 

As a result of Lemma~\ref{CRerrorIII}, we obtain the following expression for Component~(III): 
\begin{equation}
\text{Component~(III)}\ =\ \sum_{j=1}^m  \mathcal{O}\left( \sqrt{n\Delta_n} \delta_{n,j}^2 \right) +  \mathcal{O}\left( \sqrt{n\Delta_n}  (y_n-y)  \right), \quad \mbox{as $n\to\infty$}.
\label{III}
\end{equation}
Upon scaling by $(n\Delta_n)^{1/2}$, the bias of $\hat{Y}$ is influenced by the two terms on the right-hand-side of Equation~\eqref{III}. On one hand, Equation~\eqref{III} indicates that the bias of $\hat{Y}$ is of the order $\sum_{j=1}^m \mathcal{O}( \sqrt{n\Delta_n} \delta_{n,j}^2)$ as the bandwidths $\delta_{n,j}$, $j=1,2,\ldots,m$, approach zero. Consequently, to attain an asymptotically unbiased estimator $\hat{Y}$, it is necessary to select bandwidths such that $\sqrt{n\Delta_n} \delta_{n,j}^2 $ converges for all $j=1,2,\ldots,m$. The validity of this heuristic argument will be further established latter (in Theorem~\ref{CLT}). 

On the other hand, as we let $y_n=\cvv + z \sigma_Y (n \Delta_n)^{-1/2}$ and $y=\cvv$, we can observe that the last term of Equation~\eqref{III} is a constant, specifically $\mathcal{O}(z \sigma_Y)$. By carefully defining $\sigma_Y$, we can set the bias to be $z$. Consequently, the last term of Equation~\eqref{CLT:eqn10} represents the probability of a random variable being greater than or equal to $-z$. We will prove that this probability converges to $\Phi(-z)=\Phi(z)$ as $n\to\infty$, where $\Phi$ is the distribution function of standard normal distribution. This heuristic argument will also be validated in the proof of Theorem~\ref{CLT} (see Appendix~\ref{Theorem3}). 

The following lemma establishes the asymptotic normality of Component~(II). The detailed proof is included in Appendix~\ref{Lemma5}.

\begin{lemma}\label{CRerrorII}
	Suppose that Assumptions \ref{assu:dist}--\ref{assu:kernel3} hold. Then, for any sequence $y_n\in\mathcal{Y}$ such that $y_n\to y\in\mathcal{Y}$ as $n \to \infty$, we have 
	\begin{equation*}\label{proofofCLT2}
		\sqrt{n\Delta_n} \cdot \left\{ f_X(  q_\alpha )\Big( \bar{R}_n(  q_\alpha ,y_n)- \E[\bar{R}_n(  q_\alpha ,y_n)]  \Big) -g(  q_\alpha ,y)\Big( \bar{Q}_n(  q_\alpha ) - \E[\bar{Q}_n(  q_\alpha )]  \Big)  \right\}  \Rightarrow  N(0,\sigma^2)
	\end{equation*}
	as $n\to\infty$, where 
	\begin{equation}
		\sigma^2 = f_X(  q_\alpha ) g(  q_\alpha ,y)[f_X(  q_\alpha )-g(  q_\alpha ,y)]  \Big[ \int_{\R^m}  K^2( t ) \dd t\Big]^m. 
		\label{sigma}
	\end{equation}
\end{lemma}

Lemma~\ref{CRerrorII} establishes the convergence in distribution of Component~(II) to a normal distribution with variance $\sigma^2$. 
To provide an intuitive understanding of this result, we can refer to the conventional kernel estimation theory. According to this theory, we have the following results:
\begin{align*}
		\sqrt{n\Delta_n}  f_X(  q_\alpha )\Big\{ \bar{R}_n(  q_\alpha ,y_n)- \E[\bar{R}_n(  q_\alpha ,y_n)]  \Big\}   &\Rightarrow  N(0,\sigma_1^2) \quad \mbox{where $\sigma_1^2= f_X^2(  q_\alpha ) g(  q_\alpha ,y) \big[ \int_{\R^m}  K^2( t ) \dd t\big]^m$}, \\ 
		\sqrt{n\Delta_n} g(  q_\alpha ,y)\Big\{ \bar{Q}_n(  q_\alpha ) - \E[\bar{Q}_n(  q_\alpha )]  \Big\}  &\Rightarrow N(0,\sigma_2^2) \quad \mbox{where $\sigma_2^2= f_X(  q_\alpha ) g^2(  q_\alpha ,y) \big[ \int_{\R^m}  K^2( t ) \dd t\big]^m$}.
\end{align*}
Lemma~\ref{CRerrorII} essentially validates the convergence in distribution of Component~(II) to $N(0,\sigma_1^2-\sigma_2^2)$ as $n$ tends to infinity.

The following lemma established the convergence rates of Errors I--1 and I--2, and then obtains the asymptotic result of Component~(I). The detailed proof is arranged in Appendix~\ref{Lemma6}.

\begin{lemma}\label{CRerrorI}
Suppose that Assumptions \ref{assu:dist}--\ref{assu:kernel3} hold, and for all $\ell$ such that $|\ell|=1$ we have $n\delta_{n}^{2\ell} \Delta_n \to \infty$ and $\Delta_n \log\log{n} \to 0$ as $n\to\infty$. 
Then, we have 
$$
\sqrt{n\Delta_n}  \Big\{ \bar{Q}_n(  \hat{q}_\alpha ) -\bar{Q}_n(  q_\alpha )\Big\}  \to  0 \quad \mbox{in probability} \quad \mbox{as $n\to\infty$},$$ 
and for any sequence $y_n\in\mathcal{Y}$ such that $y_n\to y\in\mathcal{Y}$ as $n \to \infty$, we have 
$$
\sqrt{n\Delta_n}  \Big\{ \bar{R}_n(  \hat{q}_\alpha ,y_n) -\bar{R}_n(  q_\alpha ,y_n) \Big\}  \to  0 \quad \mbox{in probability} \quad \mbox{as $n\to\infty$}.$$
\end{lemma}

Lemma~\ref{CRerrorI} establishes the convergence of both Errors I--1 and I--2 to zero at a rate faster than $(n\Delta_n)^{-1/2}$. Consequently, Component~(I) converges to zero in probability as $n\to\infty$. This result relies on two additional restrictions imposed on the bandwidths: 
\begin{enumerate}[label={\it Restriction~(\roman*)}, align=left, leftmargin=*]
\item $n\delta_{n}^{2\ell} \Delta_n \to \infty$ for all $\ell$ such that $|\ell|=1$, 
\item $\Delta_n \log\log{n} \to 0$. 
\end{enumerate}
Remark~\ref{remark1} has already discussed Restriction~(i), which emphasizes that the bandwidths $\delta_{n,j}$, $j=1,\ldots,m$, should not converge to zero too rapidly. On the other hand, Restriction~(ii) is relatively weak and implies that the bandwidths $\delta_{n,j}$, $j=1,\ldots,m$, should not converge to zero too slowly. The subsequent remark provides further insights into these two restrictions. 

\begin{remark}\label{remark2}
By employing Taylor's expansion, we obtain the expression: 
$$
\sqrt{n\Delta_n}  \Big\{ \bar{R}_n(  \hat{q}_\alpha ,y_n) -\bar{R}_n(  q_\alpha ,y_n) \Big\} \ \approx \ \sum_{|\ell|=1}   \sqrt{2 \Delta_n \log\log{n}} \left[  \frac{\sqrt{n}}{\sqrt{2\log\log{n}}}(\hat{q}_\alpha -  q_\alpha) \right]^\ell  \D^\ell \bar{R}_{n}(  q_\alpha, y_n ). 
$$
This expression reveals that Error I--1 arises from two random factors: 
\begin{enumerate}[label={Random Factor~(\roman*)}, align=left, leftmargin=*]
\item The random nature of the empirical distribution estimator $\bar{R}(q_\alpha,y_n)$, 
\item The random nature introduced by the quantile estimator $\hat{q}_\alpha$. 
\end{enumerate}
On one hand, according to Lemma~\ref{lemmaofMV} and Remark~\ref{remark1}, Random Factor~(i) converges in probability when Restriction~(i) is imposed. On the other hand, by applying the law of the iterated logarithm for sample quantiles \citep{Serfling1980}, we conclude that the error $[n/(2\log\log n)]^{1/2}(\hat{q}_\alpha-q_\alpha)$ is asymptotically bounded almost surely. Consequently, this error tends to zero as it is multiplied by $\sqrt{2\Delta_n\log\log n}$ which tends to zero due to Restriction~(ii). Notice that, in cases where $\hat{q}_\alpha$ degenerates into a deterministic sequence converging to $q_\alpha$, as in conventional kernel estimation theory, only Random Factor~(i) remains. In such situations, Restriction~(ii) is unnecessary to obtain the results outlined in Lemma~\ref{CRerrorI}. However, when employing the sample quantile $\hat{q}_\alpha$ to estimate the quantile $q_\alpha$, it is imperative to impose Restriction~(ii) on the bandwidths to control the random error induced by $(\hat{q}_\alpha-q_\alpha)$. The same reasoning can be applied to Error I--2. For further details, please refer to Appendix~\ref{Lemma6}. 

For instance, when $\delta_{n,j}=b_n$ for $j=1,\ldots,m$, Restriction~(ii) $\Delta_n \log\log{n} \to 0$ as $n\to\infty$ can be expressed as $b_n^m \log\log n \to 0$ as $n\to\infty$. This implies that $b_n$ should converge to zero at a faster rate than $(\log\log n)^{1/m}$ and at a slower rate than $n^{-1/(2+m)}$ (see Remark~\ref{remark1}).
\end{remark}

By combining the asymptotic results of Components (I)--(III), we can now present the asymptotic result of $\hat{Y}$ in the following theorem. For a detailed proof, please refer to Appendix~\ref{Theorem3}. 

\begin{theorem}\label{CLT}
Suppose that Assumptions \ref{assu:dist}--\ref{assu:kernel3} hold, and for all $\ell$ such that $|\ell|=1$ we have $n\delta_{n}^{2\ell} \Delta_n \to \infty$ and $\Delta_n \log\log{n} \to 0$ as $n\to\infty$, and $\sqrt{n\Delta_n}\delta_{n,j}^2 \to c_j \geq 0$ as $n\to\infty$ for some constants $c_j\geq 0$, $j=1,\ldots,m$. 
Then, we have 
\begin{equation}\label{thm3:CLT0}
	\sqrt{n \Delta_n}\cdot \Bigg( \hat{Y}-\cvv- \sum_{j=1}^m \delta_{n,j}^2 \mu_j \Bigg) \Rightarrow \Normal \big(0,\sigma_Y^2 \big) \quad \mbox{as $n\to\infty$}, 
\end{equation}
where $\Normal(0,\sigma_Y^2)$ is a normal distributed random variable with variance $\sigma_Y^2$, and 
\begin{align*}
\mu_j \ &= \ \frac{ g(  q_\alpha ,\cv) \frac{\partial^2 f_X}{\partial x_j^2}(q_\alpha) - f_X(  q_\alpha ) \frac{\partial^2 g}{\partial x_j^2}(  q_\alpha ,\cv)  }{2f_X^2(  q_\alpha ) f(\cvv |  q_\alpha )}  \int_{\R} t^2 K(t) \dd t, \\
\sigma_Y^2  \ &= \ \frac{\beta(1-\beta)}{f_X(  q_\alpha ) f^2_{Y|X}(\cvv|  q_\alpha ) } \left[ \int_{\R} K^2(t) \dd t\right]^m. 
\end{align*}
\end{theorem}

Theorem~\ref{CLT} presents an intriguing result. Firstly, it demonstrates that the rate of convergence for the two-step estimator is $(n\Delta_n)^{-1/2}$, which aligns with the typical rate of convergence observed in kernel estimators of conditional quantiles. Specifically, by imposing the restriction $\sqrt{n\Delta_n} \delta_{n,j}^2 \to c_j \geq 0$, we obtain the following asymptotic result: 
$$
\sqrt{n \Delta_n}\cdot \left( \hat{Y}-\cvv \right) \Rightarrow \Normal \Bigg(\sum_{j=1}^m c_j\mu_j,\sigma_Y^2 \Bigg) \quad \mbox{as $n\to\infty$}, 
$$
As a special case where $c_j=0$, $j=1,\ldots,m$, the above asymptotic normality yields a zero mean. In other words, it implies that that $\hat{Y}-\cvv=\mathcal{O}_p\big((n\Delta_n)^{-1/2} \big)$ as $n\to\infty$. 

Secondly, it reveals that if the bandwidths satisfy Restriction~(i) $n\delta_{n}^{2\ell} \Delta_n \to \infty$ for all $\ell$ such that $|\ell|=1$, and Restriction~(ii) $\Delta_n \log\log{n} \to 0$ as $n\to\infty$, the asymptotic result~\eqref{thm3:CLT0} adopts the same form of the conventional asymptotic normality for the kernel estimators of conditional quantiles \citep[Theorem 6.3]{li2007nonparametric} as if $\hat{q}_\alpha$ were fixed and known (i.e., $\hat{q}_\alpha=q_\alpha$). As emphasized in Remark~\ref{remark2}, this arises due to the fact that when Restrictions~(i) and (ii) hold, the error induced by $(\hat{q}_\alpha-q_\alpha)$, specifically Errors I--1 and I--2, converges at a faster rate compared to the error caused by the empirical distributions, specifically Errors II--1 and II--2. Consequently, the errors stemming from $(\hat{q}_\alpha -q_\alpha)$ may be ignored. 

Thirdly, the established asymptotic normal distribution in Theorem~\ref{CLT} with $c_j=0$, $j=1,\ldots,m$, proves valuable for constructing a confidence interval for the two-step estimator $\hat{Y}$. It is noteworthy that the imposed restrictions on the bandwidths allow us to disregard the variability of $\hat{q}_\alpha$ and treat it as $q_\alpha$. Consequently, an approximate $100(1-\nu)\%$ ($0<\nu<1$) confidence interval for $\cvv$ can be obtained as follows:
\begin{equation}\label{CI}
\left( \hat{Y}- z_{1-\frac\nu2} \frac{\sigma_Y}{\sqrt{n\Delta_n}},~ \hat{Y}+ z_{1-\frac\nu2} \frac{\sigma_Y}{\sqrt{n\Delta_n}} \right), 
\end{equation}
where $z_{1-\frac\nu 2}$ is the $(1-\nu/2)$ quantile of the standard normal distribution. It should be noted that $\sigma_Y$ is dependent on $f_X(q_\alpha)$ and $f(q_\alpha,\cvv)$. In practical applications, these can be replaced with their respective kernel estimations \citep{hardle2004nonparametric}. We recognize that there may exist more sophisticated techniques for constructing confidence intervals that extend beyond the scope of this paper. We defer the investigation of these advanced approaches to future research endeavors. 

\subsection{Bandwidth Selection}\label{BS}

To achieve efficient performance in implementing the two-step estimation, it is crucial to carefully select appropriate bandwidths $\delta_{n,j}$, where $j=1,\ldots,m$. In this subsection, we present a ``heuristic" approach for the bandwidth selection. According to Theorem~\ref{CLT}, to attain a rapid rate of convergence of the standard deviation, say $\sigma_Y (n\Delta_n)^{-1/2}$, it is necessary for the bandwidths to approach zero as slowly as possible. Under the assumptions of Theorem~\ref{CLT}, the bandwidths $\delta_{n,j}\to 0$ for all $j=1,\ldots,m$ must satisfy the following simultaneous restrictions on their rates as $n\to\infty$: 
\begin{enumerate}[label={\it Restriction~(\roman*)}, align=left, leftmargin=*]
	\item $n\Delta_n \delta_{n,j}^2 \to \infty$ for all $j=1,\ldots,m$, 
	\item $\Delta_n \log\log{n} \to 0$, 
	\item $\sqrt{n\Delta_n}\delta_{n,j}^2 \to c_j \geq 0$ for all $j=1,\ldots,m$. 
\end{enumerate}
We shall only consider the case when $\delta_{n,j}=b_n$ for $j=1,\ldots,m$. In this case, Restriction~(i) can be expressed as $nb_n^{m+2} \to\infty$;  Restriction~(ii) can be expressed as $b_n^m \log\log n \to 0$; Restriction~(iii) can be expressed as $nb_n^{m+4}\to c_j^2$. Notice that, Restriction~(ii) is a relatively weak condition, which can be deduced from Restriction~(iii). By satisfying Restriction~(iii), we can adopt the bandwidth given by 
\begin{equation}\label{OptimalBandwidth}
b_n=\mathcal{O}\Big(n^{-\frac1{m+4}}\Big),
\end{equation}
and consequently, the rate of convergence becomes
\begin{equation}\label{bestrate}
\hat{Y}-\cvv = \mathcal{O}_p\left( n^{- \frac{2}{m+4}} \right) \quad \mbox{as $n\to\infty$}. 
\end{equation}
This represents the best rate of convergence that the two-step estimator $\hat{Y}$ can achieve. Notice that, when $m=1$, the above best rate of convergence becomes $\mathcal{O}_p(n^{-2/5})$. Interestingly, this best rate of convergence aligns with the best rate of convergence of the kernel estimator of quantile sensitivities introduced by \cite{LiuandHong}. Additionally, as emphasized in \cite{huanghong23}, the best rate of convergence attainable by their batching estimator of $\cvv$ is $\mathcal{O}_p(n^{-1/3})$. This observation clearly demonstrates the superior performance of our new two-step estimator over the batching estimator. 

Although the choice of bandwidth provided in Equation~\eqref{OptimalBandwidth} achieves the best rate of convergence for the standard deviation, it may result in a slower rate of convergence for the bias. Intuitively, as indicated by Theorem~\ref{CLT}, the rate of convergence of the bias is determined by $\sum \delta_{n,j}^2\mu_j$. Therefore, to attain a slower rate of convergence for the bias, it is necessary for the bandwidths to approach zero as rapidly as possible. When $\delta_{n,j}=b_n$ for $j=1,\ldots,m$, we propose utilizing the bandwidth 
\begin{equation}\label{OptimalB}
b_n=\mathcal{O}\Big(n^{-\frac1{m+4-\gamma}}\Big)
\end{equation} 
where $\gamma$ is a real-valued parameter satisfying $0<\gamma<2$. This choice satisfies Restrictions~(i)--(iii) with $c_j=0$, resulting in a rate of convergence given by
$$
\hat{Y}-\cvv = \mathcal{O}_p\left( n^{- \frac{4- \gamma}{2m+8-2\gamma}} \right) \quad \mbox{as $n\to\infty$}, 
$$
which is strictly slower than the best rate of convergence~\eqref{bestrate}. In practice, there is a bias-variance tradeoff in bandwidth selection: when choosing a small (resp., large) value of $\gamma$, we achieve a fast (resp., slow) rate of convergence for the standard deviation but a slow (resp., fast) rate of convergence for the bias. Therefore, we propose utilizing the bandwidth~\eqref{OptimalB} by carefully selecting a $\gamma\in(0,2)$. 

\section{Numerical Study}\label{App}

In this section, we examine the performance of our two-step estimator by conducting a simulation study based on the delta-gamma approximation model, an important financial model widely used for approximating complicated nonlinear portfolios. For a detailed introduction to this model, we refer readers to Section~4.1 of \cite{huanghong23} and Chapter~9 of \cite{glasserman2004monte}. Suppose we have $m+1$ portfolios that share $d$ common risk factors. According to the delta-gamma approximation, we can express their losses as follows: 
\begin{align}
X_j \ &\approx \ r_{j} + \sum_{k=1}^d \Big( p_{j,k}Z_k + q_{j,k}Z_k^2 \Big), \quad j=1,2,\ldots,m, \label{Delta-Gamma-X}\\
Y \ &\approx \ r + \sum_{k=1}^d \Big( p_{k}Z_k + q_{k}Z_k^2 \Big),  \label{Delta-Gamma-Y}
\end{align}
where the risk factors $Z_1,Z_2,\ldots,Z_d$ are independent standard normal distributed. Notice that, $p_{j,k},p_k$ represent the the first-order Greek Delta, while $q_{j,k},q_k$ represent the second-order Greek Gamma. These quantities are widely recognized as measures of price sensitivities to risk factors in financial engineering. For the subsequent analysis, we do not place strict emphasis on the errors in Equations~\eqref{Delta-Gamma-X}--\eqref{Delta-Gamma-Y}, and thus, we use the symbol ``$=$" instead of ``$\approx$".

In Section~\ref{NS-1}, we focus on the case with a univariate-quantile condition ($m=1$). We empirically validate our theoretical findings and demonstrate that our two-step estimator outperforms the batching estimator introduced by \cite{huanghong23}. In Section~\ref{NS-2}, we explore the case with a multivariate-quantiles condition where $m=2$. We emphasize the importance of studying the conditional quantile with a multivariate-quantiles condition, as it enables the measurement of systemic risk within financial networks. In contrast, the conditional quantile with a univariate-quantile condition fails to capture this systemic risk.

\subsection{Univariate-Quantile Condition}\label{NS-1}

In this subsection, we investigate a univariate case of Equations~\eqref{Delta-Gamma-X}--\eqref{Delta-Gamma-Y} with $m=1$ and $d=2$. Specifically, we consider two portfolios denoted as $X$ and $Y$, where $X$ follows a standard normal distribution, and $Y$ is composed of a quadratic form of $X$ and an independent standard normal random variable $Z$, given by 
$$
	Y = r+ p_1 X + q_1 X^2 + p_2 Z. 
$$
This simple model is commonly employed to describe a financial derivative, where $Y$ represents the loss of the derivative that depends on two risk factors. The first factor is associated with the loss of the underlying asset, denoted as $X$. Given that the price of a financial derivative often exhibits nonlinear behavior with respect to the price of its underlying asset, a quadratic form is utilized to approximate this nonlinear relationship. (The delta-gamma approximation model is based on the idea of approximating a nonlinear relationship using Taylor's expansion with a quadratic form). The second factor reflects the overall market effect, denoted as $Z$. Notice that, this model has an alternative formulation where $X$ and $Z$ are correlated; please refer to Section~5.2 of \cite{huanghong23}. However, upon careful examination, we can observe that these two formulations are equivalent. For this specific example, we can directly derive an analytical solution for $\cvv$ as follows: 
\begin{equation}\label{TrueValue}
	\cvv \ = \ r+ p_1 \Phi^{-1}(\alpha) + q_1 \big[ \Phi^{-1}(\alpha) \big]^2 +p_2 \Phi^{-1}(\beta), 
\end{equation}
where $\Phi^{-1}(\cdot)$ represents the inverse distribution function of the standard normal distribution. 

We set $\alpha=\beta=0.95$, $r=-0.1$, $p_1=0.1$, $q_1=0.3$, and $p_2=0.2$. By employing Equation~\eqref{TrueValue}, we derive the true value of $\cvv=1.21$. This result indicates that when the loss of the underlying asset $X$ reaches $q_\alpha=1.64$, we can be $95\%$ confident that the loss of the financial derivative $Y$ will not exceed $1.21$. To validate the theoretical findings of the two-step estimator, we compute the bias, standard deviation (SD), root mean square error (RMSE), and coverage probability (CP) of $95\%$ confidence intervals using the true value and $100$ replications of the estimator. Notice that, we follow the discussion in Section~\ref{asymptotic_normality} for constructing the confidence intervals. The results are presented in Tables~\ref{Table1}--\ref{Table2}.

In Table~\ref{Table1}, we present the numerical results of the two-step estimator utilizing bandwidths $\delta_{n,j}=n^{-1/4}$ (following Equation~\eqref{OptimalB} with $\gamma=1$) and employing the standard normal kernel (i.e., $K$ represents the density function of a standard normal distribution).  Additionally, we include the numerical results of the batching estimator \citep{huanghong23} with the number of batches set to $\ceil{n^{1/2}}$. The table demonstrates that as the sample size tends to infinity, the bias, SD, and RMSE of the two-step estimator approach zero. This finding aligns with the theoretical consistency discussed in Section~\ref{consistency}. Moreover, the CP converges to the nominal level of $0.95$ with increasing sample size. Upon closer examination of Table~\ref{Table1}, we observe that the two-step estimator outperforms the batching estimator, exhibiting smaller bias, SD, and RMSE. This observation provides empirical support for the theoretical findings outlined in Section~\ref{asymptotic_normality}. 

\begin{table}[ht]
	\centering
	\captionsetup{labelfont=bf}
	\caption{Performance of the Two-Step Estimator and the Batching Estimator}
	\begin{tabular}{ccccccccc}
	\toprule
	&\multicolumn{4}{c}{Kernel ($\gamma=1$)}&   \multicolumn{4}{c}{Batching} \\  \cmidrule(l){2-5} \cmidrule(l){6-9} 
	$n$ &  \multicolumn{1}{c}{Bias} & \multicolumn{1}{c}{SD} & \multicolumn{1}{c}{RMSE} & \multicolumn{1}{c}{CP} & \multicolumn{1}{c}{Bias} & \multicolumn{1}{c}{SD} & \multicolumn{1}{c}{RMSE} & \multicolumn{1}{c}{CP} \\ \midrule
	$10^2$ &  $5.18\times 10 ^{-3}$ & $3.04\times 10 ^{-1}$ &  $3.04\times 10 ^{-1}$  & $0.61$ &    $1.07$      &    $8.13\times 10^{-1}$ & $1.35$ & $0.45$   \\
	$10^3$ & $6.69\times 10^{-2}$   &  $9.85\times 10^{-2}$  & $1.19\times 10^{-1}$  & $0.69$ & $5.33\times 10^{-1}$ & $2.40\times 10^{-1}$ & $5.85\times 10^{-1}$ & $0.19$   \\
	$10^4$ &  $2.37\times 10^{-2}$  &  $4.01\times 10^{-2}$  & $4.66\times 10^{-2}$  & $0.61$ & $9.92\times 10^{-2}$ & $6.80\times 10^{-2}$ & $1.20\times 10^{-1}$ & $0.50$    \\
	$10^5$ & $9.68\times 10^{-3}$   &  $1.28\times 10^{-2}$  & $1.60\times 10^{-2}$ & $0.73$ & $7.80\times 10^{-2}$ & $3.20\times 10^{-2}$ & $8.61\times 10^{-2}$ & $0.23$    \\
	$10^6$ &  $2.58\times 10^{-3}$  &  $4.31\times 10^{-3}$  & $5.03\times 10^{-3}$  & $0.86$ & $1.52\times 10^{-2}$ & $1.37\times 10^{-2}$ & $2.05\times 10^{-2}$ & $0.82$    \\
	\bottomrule
\end{tabular}
	\label{Table1}
\end{table}

In Table~\ref{Table2}, we present the numerical results of the two-step estimator using different bandwidth selections: one is $\delta_{n,j}=n^{-1/4.9}$ (following Equation~\eqref{OptimalB} with $\gamma=0.1$), and the other is $\delta_{n,j}=n^{-1/3.1}$ (following Equation~\eqref{OptimalB} with $\gamma=1.9$). The former choice of bandwidth exhibits a slower rate of convergence towards zero. Consequently, it also demonstrates a slower rate of convergence for the bias but a faster rate of convergence for the SD compared to the latter bandwidth choice. On the other hand, the latter bandwidth selection demonstrates a faster rate of convergence for the bias but a slower rate of convergence for the SD. Given that the latter bandwidth choice yields a smaller bias, its confidence intervals will converge to the asymptotic regime~\eqref{CI} more rapidly, resulting in improved CP. This observation provides empirical support for the discussion in Section~\ref{BS}. 

\begin{table}[ht]
	\centering
	\captionsetup{labelfont=bf}
	\caption{Performance of the Two-Step Estimator for Different Bandwidth Selections}
	\begin{tabular}{ccccccccc}
		\toprule
		&\multicolumn{4}{c}{Kernel ($\gamma=0.1$)}&   \multicolumn{4}{c}{Kernel ($\gamma=1.9$)} \\  \cmidrule(l){2-5} \cmidrule(l){6-9} 
		$n$ &  \multicolumn{1}{c}{Bias} & \multicolumn{1}{c}{SD} & \multicolumn{1}{c}{RMSE} & \multicolumn{1}{c}{CP} & \multicolumn{1}{c}{Bias} & \multicolumn{1}{c}{SD} & \multicolumn{1}{c}{RMSE} & \multicolumn{1}{c}{CP} \\ \midrule
		$10^2$ &  $2.87\times 10 ^{-2}$ & $2.89\times 10 ^{-1}$ &  $2.90\times 10 ^{-1}$  & $0.67$ &    $6.70\times 10 ^{-2}$      &    $2.92\times 10 ^{-1}$ & $2.99\times 10 ^{-1}$ & $0.59$   \\
		$10^3$ &  $1.11\times 10 ^{-1}$  &  $1.01\times 10 ^{-1}$  & $1.50\times 10 ^{-1}$ & $0.53$ & $2.14\times 10 ^{-2}$  & $1.10 \times 10 ^{-1}$ & $1.12\times 10 ^{-1}$ & $0.76$  \\
		$10^4$ & $5.72\times 10 ^{-2}$   &  $3.68\times 10 ^{-2}$  & $6.80\times 10 ^{-2}$ & $0.31$ & $4.10\times 10^{-3}$ & $4.27\times 10^{-2}$ & $4.29\times 10^{-2}$ & $0.82$  \\
		$10^5$ & $2.80\times 10 ^{-2}$   & $1.17\times 10 ^{-2}$   & $3.04\times 10 ^{-2}$ & $0.18$ & $3.46\times 10^{-4}$ & $1.65\times 10^{-2}$ & $1.65\times 10^{-2}$ & $0.84$  \\
		$10^6$ & $1.12\times 10 ^{-2}$   &  $3.65\times 10 ^{-3}$  & $1.17\times 10 ^{-2}$ & $0.14$ & $3.41\times 10^{-5}$ & $6.26\times 10^{-3}$ & $6.26\times 10^{-3}$ &  $0.95$  \\
		\bottomrule
	\end{tabular}
	\label{Table2}
\end{table}

Figure~\ref{fig1} illustrates the rate of convergence of the two-step estimator employing the bandwidth $\delta_{n,j}=n^{-1/5}$ (as indicated in Equation~\eqref{OptimalBandwidth}). Through simulating data for various sample sizes $n$, we capture different RMSEs. Subsequently, by applying logarithmic transformation and conducting linear regression, we derive the rate of convergence of RMSE. Figure~\ref{fig1} demonstrates that the empirical rate of convergence of RMSE is approximately $n^{-0.3983}$, closely aligning with the theoretical rate of $\mathcal{O}_p(n^{-2/5})$ discussed in Section~\ref{BS}. Additionally, it establishes that the two-step estimator achieves a superior rate of convergence compared to the batching estimator developed by \cite{huanghong23}, which exhibits the best rate of $\mathcal{O}_p(n^{-1/3})$. 

\begin{figure}[ht]
	\captionsetup{labelfont=bf}
	\caption{The Rate of Convergence of the Two-Step Estimator with Bandwidth~\eqref{OptimalBandwidth}}
	\centering    
	\includegraphics[scale = 0.36]{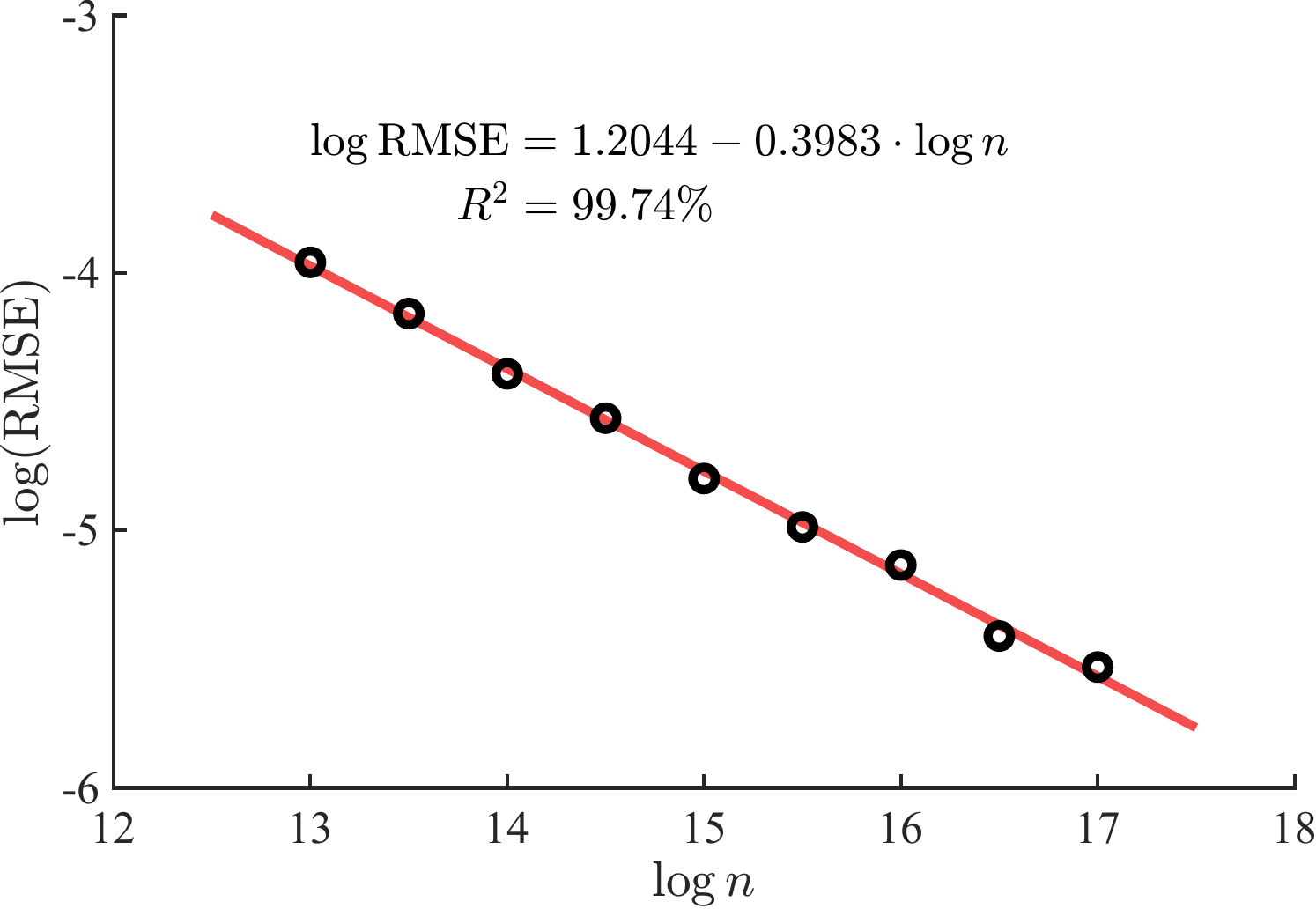}
	\captionsetup{labelfont=bf}
	\caption*{\footnotesize \it Note: The figure illustrates the rate of convergence of the two-step estimator, which can achieve $\mathcal{O}_p(n^{-2/5})$, in alignment with the theoretical result.}
	\label{fig1}
\end{figure}

\subsection{Multivariate-Quantiles Condition}\label{NS-2}

In this subsection, we investigate a multivariate case of Equations~\eqref{Delta-Gamma-X}--\eqref{Delta-Gamma-Y} with $m=2$ and $d=2$. This serves as a simple example of a financial market comprising three financial institutions and two risk factors. Let us denote the losses of these institutions as $X_1, X_2, Y$, which are contingent upon two independent risk factors denoted as $Z_1, Z_2$ and following a standard normal distribution. More precisely, we define the model as follows: 
\begin{align*}
	X_1 \ &= \ -0.15 +  0.6\cdot Z_1 + 0.8\cdot Z_1^2  -0.2\cdot Z_2 -0.2\cdot Z_2^2 , \\
	X_2 \ &= \ -0.12 -0.2\cdot Z_1 -0.2 \cdot Z_1^2 +0.8\cdot Z_2 + 0.6\cdot Z_2^2, \\
	Y \ &= \ -0.10 + 0.2\cdot Z_1 +  0.2\cdot Z_1^2 +  0.1\cdot Z_2 +  0.3\cdot Z_2^2. 
\end{align*}
We interpret the model in the following manner: $X_1, X_2, Y$ represent the loss variables for three financial institutions, while $Z_1, Z_2$ denote the loss variables for two securities held by these institutions. The correlation of risk arises since they are holding common securities. 

An essential task of the risk manager at institution $Y$ is to evaluate the magnitude of the risk associated with $Y$, taking into consideration the correlation with two other institutions $X_1$ and $X_2$. The risk manager can measure the risk of $Y$ with the following different manners:
\begin{enumerate}[label={\it Measure~(\roman*)}, align=left, leftmargin=*]
\item $\cvv$ defined in this paper:  $\pr\big\{ Y\leq \cvv \big| X_1=q_{1,\alpha_1}, X_2=q_{2,\alpha_2} \big\}=\beta$, 
\item ${\rm VaR}_\beta(Y)$ defined without conditions:  $\pr\big\{ Y\leq {\rm VaR}_\beta(Y) \big\}=\beta$, 
\item $\cvv$ with a univariate condition:  $\pr\big\{ Y\leq \cvv \big| X_1=q_{1,\alpha_1} \big\}=\beta$, 
\item $\cvv$ with a univariate condition:  $\pr\big\{ Y\leq \cvv \big| X_2=q_{2,\alpha_2} \big\}=\beta$, 
\item $\cvv$ with an aggregate condition: $\pr\big\{ Y\leq \cvv \big| X_1+X_2=q_{1,\alpha_1}+q_{2,\alpha_2} \big\}=\beta$. 
\end{enumerate}
It is important to emphasize that an analytical form for Measures~(i)--(v) is \textit{not} available within this model. Existing simulation-based methods \citep{huanghong23} can estimate Measures~(iii)--(v), but they are incapable of estimating Measure~(i).

To obtain estimates for Measure~(i) and Measures~(iii)--(v), we employ the two-step estimation approach. Specifically, we generate a large sample of $X_1,X_2,Y$ with $n=10^6$ observations to obtain point estimates. Subsequently, we conduct $100$ replications and compute the mean. For Measure~(i), we adopt a bandwidth of $\delta_{n,j}= n^{-1/5}$ (following Equation~\eqref{OptimalB} with $\gamma=1$), while for Measures~(iii)--(v), we use a bandwidth of $\delta_{n,j}= n^{-1/4}$ (following Equation~\eqref{OptimalB} with $\gamma=1$). Note that, we employ the standard normal kernel for the above kernel estimations. The estimation of Measure~(ii) is based on order statistics, as outlined in \cite{Serfling1980}. 

Table~\ref{Table3} showcases the estimation results for Measures~(i)--(v) at various $\alpha,\beta$ values. It can be regarded as a report submitted to the CEO of institution $Y$ by its risk manager, providing a comprehensive list of different scenarios.  In order to evaluate diversity, the risk manager systematically considers a range of possibilities involving different values of $\alpha$ and $\beta$. 

Table~\ref{Table3} (the fourth column) presents the results of $\cvv$ defined in this paper. Utilizing these findings, we can evaluate the increase in value-at-risk for institution $Y$ during crises that impact $X_1$ and $X_2$, as expressed by
$$
\Delta \cvv = \cv_{(\alpha_1,\alpha_2,\beta)}-\cv_{(0.5,0.5,\beta)}
$$
For instance, we observe that $\Delta \cv_{(0.95,0.95,0.95)}= 2.485-0.321=2.164$, indicating that the value-at-risk of institution $Y$ would experience a rise of $2.164$ when the losses of institutions $X_1$ and $X_2$ shift from the median state ($X_1=q_{1,0.5}$ and $X_2=q_{2,0.5}$) to the crisis state ($X_1=q_{1,0.95}$ and $X_2=q_{2,0.95}$). Considering different values of $\alpha_1,\alpha_2,\beta$ yields varying outcomes. For example, the observation
$$
\cv_{(0.8,0.95,\beta)}>\cv_{(0.95,0.8,\beta)}
$$ 
implies that the value-at-risk of $Y$ may exhibit a greater sensitivity to $X_2$ than to $X_1$. Consequently, we could infer that $X_2$ poses a higher risk contagion to $Y$ than $X_1$. 

Table~\ref{Table3} also highlights a significant disparity between Measure~(i) and Measures~(ii)--(v). In situations where the risk manager acknowledges the potential propagation of risks from both institutions $X_1$ and $X_2$ to $Y$, Measure~(i) quantifies the extent of financial contagion originating from both $X_1$ and $X_2$. However, when considering only a partial examination of either $X_1$ or $X_2$, Measures~(iii) and (iv) respectively measure the financial contagion from either $X_1$ or $X_2$. Consequently, Measures~(iii)--(iv) can be regarded as node-to-node measures of systemic risk, unable to capture the simultaneous impact of risks associated with both $X_1$ and $X_2$ within the financial networks.

In contrast, when employing an aggregate approach of $X_1+X_2$, Measure~(v) quantifies the overall financial contagion of $X_1+X_2$. It becomes indistinguishable whether the financial contagion originates from $X_1$ or $X_2$. For instance, within the aggregate condition, we are unable to differentiate among the following three events: $\{X_1=q_{1,\alpha_1}, X_2=q_{2,\alpha_2}\}$, $\{X_1=q_{2,\alpha_2}, X_2=q_{1,\alpha_1}\}$, and $\{X_1=(q_{1,\alpha_1}+q_{2,\alpha_2})/2, X_2=(q_{1,\alpha_1}+q_{2,\alpha_2})/2\}$, even though $q_{1,\alpha_1}$ and $q_{2,\alpha_2}$ may differ considerably. 

Conversely, if the risk manager disregards any consideration related to $X_1$ and $X_2$, that is, focusing solely on assessing the risk of $Y$ in isolation, Measure~(ii) quantifies the value-at-risk of $Y$. The outcomes obtained from Measures~(ii)--(v) in Table~\ref{Table3}  significantly deviate from the value of Measure~(i), thus emphasizing the potential for drawing erroneous conclusions regarding systemic risk when crucial considerations of correlated institutions are omitted.

\begin{table}[ht]
	\centering
    \captionsetup{labelfont=bf}
	\caption{The Estimation Results for Different Measures of Risk Given in Measures~(i)--(v)}
	\label{tab:three-line-chart}
	\begin{tabular}{cccccccc}
		\toprule
		~$\alpha_1$~ & ~$\alpha_2$~ & ~$\beta$~ & ~Measure~(i)~ & ~Measure~(ii)~ & ~Measure~(iii)~ & ~Measure~(iv)~ & ~Measure~(v)~ \\
		\midrule
		$0.50$ & $0.50$ &  & $0.047$ &  & $0.210$ & $0.189$ & $0.088$ \\
		$0.80$ & $0.80$ &  & $0.816$ &  & $0.505$ & $0.711$ & $1.083$ \\
		$0.80$ & $0.95$ & $0.80$ & $1.734$ & $0.608$ & $0.505$ & $1.316$ & $2.053$ \\
		~$0.95$~ & $0.80$ &  & $1.413$ &  & $0.878$ & $0.711$ & $2.157$ \\
		$0.95$ & $0.95$ &  & $2.337$ &  & $0.878$ & $1.316$ & $2.999$ \\ \midrule
		$0.50$ & $0.50$ &  & $0.321$ &  & $0.738$ & $0.561$ & $0.422$ \\
		$0.80$ & $0.80$ &  & $1.347$ &  & $1.101$ & $1.080$ & $1.170$ \\
		$0.80$ & $0.95$ & $0.95$ & $1.829$ & $1.336$ & $1.101$ & $1.878$ & $2.290$ \\
		$0.95$ & $0.80$ &  & $1.761$ &  & $1.488$ & $1.080$ & $2.418$ \\
		$0.95$ & $0.95$ &  & $2.485$ &  & $1.488$ & $1.878$ & $3.569$ \\
		\bottomrule
	\end{tabular}
\label{Table3}
\end{table}

\section{Conclusions}\label{conclusion}

In this paper, we have investigated the estimation problem of the systemic risk measure $\cv$ under a multivariate-quantiles condition. We have highlighted the limitations of existing model-based and simulation-based methods within this context. To address these limitations, we have introduced a two-step nonparametric estimation method that effectively incorporates the flexibility of Monte-Carlo simulation. Our proposed approach is capable of handling the multivariate-quantiles condition and is based on flexible model specifications. We have established the consistency and asymptotic normality of our two-step estimator. The derived asymptotic results provide valuable insights into bandwidth selection and the rate of convergence. Additionally, we have validated our theoretical findings through numerical experiments, which have confirmed the strong performance of our two-step estimator.

\appendix

\section{Proof of Lemma \ref{lemmaofMV}}\label{Lemma1}

We first define some notations that will be used in the proof. Recall that $x=(x_1,x_2,\ldots,x_m)^\top$ and $\delta_n=(\delta_{n,1},\delta_{n,2},\ldots, \delta_{n,m})^\top$. We denote 
\begin{equation*}
	x \odot \delta_{n} \ =\ \big( x_1 \delta_{n,1},~ x_2\delta_{n,2}, ~\ldots, ~x_m\delta_{n,m} \big)^\top \quad \text{and} \quad 
	x \oslash \delta_{n} \ =\ \left( \frac{x_1}{\delta_{n,1}},~ \frac{x_2}{\delta_{n,2}}, ~\ldots, ~\frac{x_m}{\delta_{n,m}} \right)^\top, 
\end{equation*}
where the symbols $\odot$ and $\oslash$ denote the element-wise product and element-wise division operations, respectively. In calculus, these operations are commonly referred to as the Hadamard product and Hadamard division, respectively. We also denote 
$$
\dd x  \ =\ \dd x_1 \dd x_2 \cdots \dd x_m \quad \text{and} \quad 
\dd \left( x \oslash \delta_{n} \right) \ =\  \dd \left(\frac{x_1}{\delta_{n,1}} \right) \dd \left( \frac{x_2}{\delta_{n,2}} \right) \cdots \dd \left( \frac{x_m}{\delta_{n,m}} \right). 
$$
To prove Lemma~\ref{lemmaofMV}, we first prove the following Lemma~\ref{Bochner}. Recall that $\Delta_n=\prod_{j=1}^m \delta_{n,j}$. 

\begin{lemma}\label{Bochner}
Suppose that $ H(x):\R^m\to\R$ is a multivariate function satisfying $\sup_{x \in\R^m} |  H(x) | <\infty$, $\int_{\R^m} | H(x) | \dd x <\infty$, and $\lim_{\|x\| \to\infty} \| x\| \big| H(x) \big| = 0$. Let $G(\cdot):\R^m\to\R$ be a multivariate function satisfy $\int_{\R^m} |G(x)| \dd x<\infty$. Let $\delta_{n,j}$ be sequences of positive constants satisfying $\lim_{n\to\infty}\delta_{n,j} =0 $ for all $j=1,\ldots,m$. Then, at every point $z $ of continuity of $G$, we have
$$
\lim_{n\to\infty}  \frac{1}{\Delta_n} \int_{\R^m}  H\left( x \oslash \delta_{n} \right) G(z -x) \dd x \ = \ G(z )\int_{\R^m} H(x) \dd x. 
$$
\end{lemma}

\begin{proof}
For any $\delta>0$, we have 
\begin{eqnarray}
\lefteqn{
	\left| \frac{1}{\Delta_n} \int_{\R^m}  H\left( x \oslash \delta_{n} \right) G(z -x) \dd x - G(z )\int_{\R^m} H(x) \dd x \right| } \nonumber \\
&=& \left| \frac{1}{\Delta_n} \int_{\R^m} \big[G(z -x) - G(z )\big] H\left( x \oslash \delta_{n} \right)  \dd x \right| \nonumber \\
&\leq&   \max_{\|x\|< \delta} \big| G(z -x)-G(z ) \big| \int_{\|x\|< \delta} \big| H\left( x \oslash \delta_{n} \right) \big| \dd \left(  x \oslash \delta_{n} \right) \nonumber \\
&&
~~~~~~ +\ \int_{\|x\|\geq \delta} \big| G(z -x) \big|   \big| H\left( x \oslash \delta_{n} \right)\big| \dd \left( x \oslash \delta_{n} \right) \ + \  \int_{\|x\|\geq \delta}  \big| G(z ) \big|  \big| H\left( x \oslash \delta_{n} \right) \big| \dd \left(  x \oslash \delta_{n} \right)   \nonumber \\
&\leq& \max_{\|x\|<\delta } \big| G(z -x)-G(z ) \big| \int_{ \R^m} \left| H\left( x \right) \right| \dd x \nonumber \\
&&
~~~~~~ +\ \frac{\|\delta_n\|}{\delta} \int_{\R^m} |G(x)| \dd x  \sup_{\| x \| \geq \delta/\|\delta_n\| } \Big\{  \|x\| \big| H( x ) \big| \Big\} \ + \ \big| G(z ) \big| \int_{\|x\|\geq \delta/\|\delta_n\|}  \left| H\left( x \right) \right| \dd x  \nonumber
\end{eqnarray}
which goes to zero by sending $n\to\infty$ and then sending $\delta\to 0$. Therefore, we conclude the proof. 
\end{proof}

Notice that, when $m=1$, Lemma \ref{Bochner} degenerates to the Bochner's lemma (see, e.g., \citealt{parzen1962estimation}, \citealt{bochner2005harmonic}, \citealt{LiuandHong}). However, Lemma \ref{Bochner} provides a stronger result for the multivariate case (i.e., $m\geq 1$). In light of Lemma \ref{Bochner}, we prove Lemma \ref{lemmaofMV} as follows. 

\begin{proof}{(Proof of Lemma \ref{lemmaofMV})} Due to the similar nature of the proof for Equations~\eqref{lem1:eqn:10} and \eqref{lem1:eqn:3} to that of Equation~\eqref{lem1:eqn:1}, and the similar nature of the proof for Equations~\eqref{lem1:eqn:20} to that of Equation~\eqref{lem1:eqn:2}, we focus solely on establishing the validity of Equations~\eqref{lem1:eqn:1} and \eqref{lem1:eqn:2} in the subsequent analysis.

Let us define $K_{n,j}(t) = K\big((t-x_j)/\delta_{n,j}\big)$ and $K_{n,j}^{(\ell_j)}(t) =   K^{(\ell_j)}\big((t-x_j)/\delta_{n,j}\big)/\delta_{n,j}^{\ell_j}$ for $t\in\R$. 
Then, we have 
\begin{eqnarray}
\lefteqn{\E \left[\frac{1}{n \Delta_n } \sum_{i=1}^n \D^\ell W_{n,i}(  q_\alpha ) \cdot \mathbb{I}\{Y_i\leq y_n\} \right] \ = \ \frac{1}{\Delta_n } \E \left[ \D^\ell W_{n,1}(  q_\alpha ) \cdot \mathbb{I}\{Y_1 \leq y_n\} \right] } \nonumber \\
&=& \frac{1}{\Delta_n  } \int_{\R} \cdots \int_{\R} ~ \prod_{j=1}^m  K^{(\ell_j)}_{n,j} ( q_{j,\alpha_j} ) \cdot g(x,y_n) ~ \dd x_1 \cdots \dd x_m \nonumber \\
&=& \int_{\R} ~ \frac{1}{\delta_{n,m}} K^{(\ell_m)}_{n,m} ( q_{m,\alpha_m} ) \cdots \int_{\R}~ \frac{1}{\delta_{n,1}} K^{(\ell_1)}_{n,1} ( q_{1,\alpha_1} ) \cdot g(x,y_n) ~ \dd x_1 \cdots \dd x_m. ~~ \label{lem1:eqn1}
\end{eqnarray}

Based on the assumption that $|\ell|\leq 1$, we can identify two possible scenarios: (i) $\ell_1=\ell_2=\cdots=\ell_m=0$; (ii) there exists a value of $j=1,2,\ldots,m$ such that $\ell_j=1$ and $\ell_i=0$ for $i\neq j$. In the case of scenario (i), we have 
\begin{eqnarray}
	\lefteqn{\E \left[ \frac{1}{n \Delta_n } \sum_{i=1}^n \D^\ell W_{n,i}(  q_\alpha ) \cdot \mathbb{I}\{Y_i\leq y_n\} \right]} \nonumber \\
	&=&  \frac{1}{\Delta_{n}} \int_{\R^m}  \prod_{j=1}^m K_{n,j}(q_{j,\alpha_j}) \cdot  \left[ g(x,y_n) - g(x,y) \right] \dd x + \frac{1}{\Delta_{n}} \int_{\R^m}  \prod_{j=1}^m K_{n,j}(q_{j,\alpha_j}) \cdot  g(x,y)  \dd x.  \label{lem1:eqn30}
\end{eqnarray}
By (iii) of Assumption~\ref{assu:dist} and Assumption~\ref{assu:kernel1}, we have the first term of Equation \eqref{lem1:eqn30} vanishes as $n\to\infty$. By (ii) and (iv) of Assumption~\ref{assu:dist}, we have $ g(x,y)$ is continuous at $x=  q_\alpha $ and $\int_{\R^m}| g(x,y)| \dd  x<\infty$. By Assumption~\ref{assu:kernel1}, we have $\sup_{t \in\R^m} | \prod_{j=1}^m K_{n,j}(t_j) |< \infty$, $\int_{ \R^m} |\prod_{j=1}^m K_{n,j}(t_j)| \dd t <\infty $, and $\lim_{\|t\| \to\infty} \|t\|$ $| \prod_{j=1}^m K_{n,j}(t_j)| = 0$. Then, by Lemma \ref{Bochner}, the second term of Equation \eqref{lem1:eqn3} converges to $ g(  q_\alpha ,y )$ as $n\to\infty$.  Thus, we obtain Equation \eqref{lem1:eqn:1}.

In the case of scenario (ii), we can proceed without loss of generality by assuming $\ell_1=1$ and $\ell_2=\cdots=\ell_m=0$. By employing the technique of integration by parts, we obtain 
\begin{eqnarray}
\lefteqn{  \int_{\R}~ \frac{1}{\delta_{n,1}} K^{(\ell_1)}_{n,1} ( q_{1,\alpha_1} ) \cdot g(x,y_n) ~ \dd x_1 \ = \ \int_{\R}~ \frac{1}{\delta_{n,1}^{\ell_1+1}} K^{(\ell_1 )}\left( \frac{q_{1,\alpha_1} - x_1}{\delta_{n,1}} \right) g(x,y_n) ~ \dd x_1 } \nonumber \\
	&=& -\left.\frac{1}{\delta_{n,1}^{\ell_1}} K^{(\ell_1 - 1)}\left( \frac{q_{1,\alpha_1} - x_1}{\delta_{n,1}} \right) g(x,y_n)\right|_{-\infty} ^{\infty} + \frac{1}{\delta_{n,1}^{\ell_1}} \int_{\R} K^{(\ell_1 -1)}\left( \frac{q_{1,\alpha_1} - x_1}{\delta_{n,1}} \right) \frac{\partial g}{\partial x_1}(x,y_n) ~ \dd  x_1 \nonumber \\
	&=& \frac{1}{\delta_{n,1}^{\ell_1}} \int_{\R} K^{(\ell_1 -1)}\left( \frac{q_{1,\alpha_1} - x_1}{\delta_{n,1}} \right) \frac{\partial g}{\partial x_1}(x,y_n) ~ \dd  x_1,  \label{lem1:eqn2} 
\end{eqnarray}
where the last equation holds by Assumption~\ref{assu:kernel1}. Substituting Equation \eqref{lem1:eqn2} into Equation \eqref{lem1:eqn1}, we have 
\begin{eqnarray}
\lefteqn{\E \left[ \frac{1}{n \Delta_n } \sum_{i=1}^n \D^\ell W_{n,i}(  q_\alpha ) \cdot \mathbb{I}\{Y_i\leq y_n\} \right]} \nonumber \\
&=&  \frac{1}{\Delta_{n}} \int_{\R^m}  \prod_{j=1}^m K_{n,j}(q_{j,\alpha_j}) \cdot  \D^\ell g(x,y_n)   \dd x \nonumber \\
&=& \frac{1}{\Delta_{n}} \int_{\R^m}  \prod_{j=1}^m K_{n,j}(q_{j,\alpha_j}) \cdot  \left[ \D^\ell g(x,y_n) -\D^\ell g(x,y) \right] \dd x + \frac{1}{\Delta_{n}} \int_{\R^m}  \prod_{j=1}^m K_{n,j}(q_{j,\alpha_j}) \cdot   \D^\ell g(x,y)  \dd x.~~~~~~ \label{lem1:eqn3}
\end{eqnarray}
By (iii) of Assumption~\ref{assu:dist} and Assumption~\ref{assu:kernel1}, we have the first term of Equation \eqref{lem1:eqn3} vanishes as $n\to\infty$. By (ii) and (iv) of Assumption~\ref{assu:dist}, we have $\D^\ell g(x,y)$ is continuous at $x=  q_\alpha $ and $\int_{\R^m}|\D^\ell g(x,y)| \dd  x<\infty$. By Assumption~\ref{assu:kernel1}, we have $\sup_{t \in\R^m} | \prod_{j=1}^m K_{n,j}(t_j) |< \infty$, $\int_{ \R^m} |\prod_{j=1}^m K_{n,j}(t_j)| \dd t <\infty $, and $\lim_{\|t\| \to\infty} \|t\|$ $| \prod_{j=1}^m K_{n,j}(t_j)| = 0$. Then, by Lemma \ref{Bochner}, the second term of Equation \eqref{lem1:eqn3} converges to $\D^\ell g(  q_\alpha ,y )$ as $n\to\infty$.  Thus, we obtain Equation \eqref{lem1:eqn:1}. 

Now, we prove Equation \eqref{lem1:eqn:2}. 
By a direct consideration, we have 
\begin{eqnarray}
\lefteqn{ n \delta_{n}^{2\ell } \Delta_n \cdot \Var \left( \frac{1}{n \Delta_n  } \sum_{i=1}^n \D^\ell W_{n,i}(  q_\alpha ) \cdot \mathbb{I}\{Y_i\leq y_n\} \right)} \nonumber \\ 
	&=&   \frac{\delta_{n}^{2\ell }}{ \Delta_n } \cdot \E \left[ [ \D^\ell W_{n,1}(  q_\alpha )]^2 \cdot \mathbb{I}\{Y_i\leq y_n\} \right] \ -\ \delta_{n}^{2\ell } \Delta_n \cdot \left( \E \left[ \frac{1}{\Delta_n  } \D^\ell W_{n,1}(  q_\alpha ) \cdot \mathbb{I}\{Y_i\leq y_n\} \right]\right)^2. \label{lem1:eqn:0}
\end{eqnarray}
By Equation \eqref{lem1:eqn:1}, we have the second term of Equation \eqref{lem1:eqn:0} vanishes as $n\to\infty$. 
By (ii) and (iv) of Assumption~\ref{assu:dist}, we have $g(x,y)$ is continuous at $x=  q_\alpha $ and $\int_{\R^m}|g(x,y)| \dd  x<\infty$. By Assumption~\ref{assu:kernel1}, we have $\sup_{t \in\R^m} | \D^\ell \prod_{j=1}^m K_{n,j}(t_j) |^2 < \infty$, $\int_{ \R^m} |\D^\ell \prod_{j=1}^m  K_{n,j}(t_j)|^2 \dd t <\infty $, and $\lim_{\|t\| \to\infty} \|t\| | \D^\ell \prod_{j=1}^m K_{n,j}(t_j)|^2$ $= 0$. 
Then, it follows that  
\begin{eqnarray}
\lefteqn{ \frac{\delta_{n}^{2\ell} }{\Delta_n} \E \left[ [ \D^\ell W_{n,1}(  q_\alpha )]^2 \cdot \mathbb{I}\{Y_i\leq y_n\} \right] } \nonumber \\
&=&  \frac{1 }{\Delta_n } \int_{\R^m} ~ \Bigg[ \prod_{j=1}^m  K^{(\ell_j)} \left( \frac{q_{j,\alpha_j} -x_j}{\delta_{n,j}} \right) \Bigg]^2  g(x,y_n) ~ \dd x \nonumber \\
&=&  \frac{1 }{\Delta_n } \int_{\R^m} ~ \Bigg[ \prod_{j=1}^m  K^{(\ell_j)} \left( \frac{q_{j,\alpha_j} -x_j}{\delta_{n,j}} \right) \Bigg]^2  [g(x,y_n)-g(x,y)] ~ \dd x \label{lem7:eqn24} \\
&&  ~~~~~~~~~~~~~~~~~~~~~~~~~~~~~~~~~~~~~~ +\frac{1 }{\Delta_n } \int_{\R^m} ~ \Bigg[ \prod_{j=1}^m  K^{(\ell_j)} \left( \frac{q_{j,\alpha_j} -x_j}{\delta_{n,j}} \right) \Bigg]^2  g(x,y) ~ \dd x  \label{lem7:eqn245} \\
&\to&  g(  q_\alpha , y) \int_{\R^m} \Bigg[ \D^\ell \prod_{j=1}^m  K ( x_j )  \Bigg]^2 \dd  x \quad \mbox{as $n\to\infty$,} \label{lem7:eqn25}
\end{eqnarray}
where the limit of Equation~\eqref{lem7:eqn24} tends to zero as $n$ approaches infinity due to the fulfillment of condition~(iii) in Assumption~\ref{assu:dist} and Assumption~\ref{assu:kernel1}, and Equation~\eqref{lem7:eqn245} converges to Equation~\eqref{lem7:eqn25} as $n$ tends to infinity based on the application of Lemma~\ref{Bochner}. By substituting Equation~\eqref{lem7:eqn25} into Equation~\eqref{lem1:eqn:0}, we derive Equation~\eqref{lem1:eqn:2}. Therefore, we conclude the proof.
\end{proof}

\section{Proof of Lemma \ref{2error0}}\label{Lemma2}

\begin{proof}
	By applying Chebyshev's inequality (Theorem 1.6.4 of \citealt{durrett2019probability}), for any $\varepsilon>0$, we have 
	$$
	\pr \left\{ | \bar{R}_n(   q_\alpha  ,y) - \E[ \bar{R}_n(   q_\alpha  ,y) ] | \geq \varepsilon \right\} \leq \frac{ \Var\big( \bar{R}_n(   q_\alpha  ,y)\big)}{\varepsilon^2}. 
	$$
	Considering Lemma \ref{lemmaofMV} and the assumption that $n \Delta_n \to \infty$ as $n$ tends to infinity, it follows that $\Var\big( \bar{R}_n(   q_\alpha  ,y)\big) \to 0$ as $n$ approaches infinity. Consequently, we can conclude that Error II--1 converges to zero as $n$ tends to infinity. By a similar line of reasoning, we can demonstrate that Error II--2 also converges to zero as $n$ tends to infinity. Therefore, we conclude the proof.
\end{proof}

\section{Proof of Lemma \ref{2error}}\label{Lemma3}

\begin{proof}
By Taylor's expansion, we have
\begin{eqnarray}
	\lefteqn{ \bar{R}_n(   \hat{q}_\alpha  ,y) -  \bar{R}_n(   q_\alpha  ,y) } \label{lem2:eqn1} \\
	&=& {1\over n\Delta_n} \sum_{i=1}^{n} W_{n,i}(   \hat{q}_\alpha  ) \mathbb{I}\{Y_i\leq y\} - {1\over n\Delta_n} \sum_{i=1}^{n} W_{n,i}(  q_\alpha  ) \mathbb{I}\{Y_i\leq y\}  \nonumber\\
	& =& \sum_{|\ell|=1} \frac{1}{\ell !}  \left(  \hat{q}_\alpha -  q_\alpha \right)^\ell \frac1{n \Delta_n } \sum_{i=1}^n \D^\ell W_{n, i}(  q_\alpha )  \mathbb{I}\{Y_i\leq y\} \label{lem2:1term}\\
	 && ~~~~~~~~~~~~~~~~~~~~~~~~~~~~~~~~~~~~~~~ + \sum_{|\ell|=1} o_p\left( \left(  \hat{q}_\alpha -  q_\alpha \right)^\ell \frac1{n  \Delta_n } \sum_{i=1}^n \D^\ell W_{n, i}(  q_\alpha )  \mathbb{I}\{Y_i\leq y\} \right). ~~~~ \label{lem2:2term}
\end{eqnarray}
By Lemma \ref{lemmaofMV}, for any $\ell$ such that $|\ell|=1$, we have
\begin{align}
	 \E\left[ \frac{1}{n \Delta_n } \sum_{i=1}^n \D^\ell W_{n, i}(  q_\alpha )  \mathbb{I}\{Y_i\leq y\}\right] \ &\to \  \D^\ell g(q_\alpha,y), \label{lem2:var0} \\
	  \Var\left(  \frac{1}{n  \Delta_n } \sum_{i=1}^n \D^\ell W_{n, i}(  q_\alpha )  \mathbb{I}\{Y_i\leq y\} \right) \ &=\  \mathcal{O} \left( \frac{1}{n\delta_{n}^{2\ell} \Delta_n} \right), \label{lem2:var} 
\end{align}
as $n\to\infty$. By the assumption $n\delta_{n}^{2\ell} \Delta_n \to \infty$ as $n\to\infty$, we have Equation \eqref{lem2:var} vanishes as $n\to\infty$. 
Then, by Chebyshev's inequality (Theorem 1.6.4 of \citealt{durrett2019probability}), for $|\ell|=1$, we have 
\begin{equation}
 \frac{1}{n \Delta_n } \sum_{i=1}^n \D^\ell W_{n, i}(  q_\alpha )  \mathbb{I}\{Y_i\leq y\} \ \to \  \D^\ell g(q_\alpha,y), 
\label{LLN1}
\end{equation}
in probability as $n\to\infty$. By continuous mapping theorem (Section 1.7 of \citealt{Serfling1980}) and Slutsky's theorem (Section 1.5.4 of \citealt{Serfling1980}), we have $\left(  \hat{q}_\alpha -  q_\alpha \right)^\ell $ converges to zero almost surely as $n\to\infty$. Thus, by Slutsky's theorem, we have Equation \eqref{lem2:1term} converges to zero in probability as $n\to\infty$. Similarly, we have Equation \eqref{lem2:2term} converges to zero in probability as $n\to\infty$. Thus, we conclude the proof of Error~I--1. Similarly, we can prove that for Error~I--2. Therefore, we conclude the proof. 
\end{proof}

\section{Proof of Theorem \ref{thm34}}\label{Theorem2}

\begin{proof}
The proof is based on the proof presented in Theorem~2 of \cite{huanghong23}. For any sufficiently small $\tilde{\varepsilon}>0$, it holds that both ${\rm CoVaR}_{\alpha,\beta} -\tilde{\varepsilon}$ and ${\rm CoVaR}_{\alpha,\beta} +\tilde{\varepsilon}$ belong to $\mathcal{Y}$. According to the definition~\eqref{defCoVaRGeneral} of ${\rm CoVaR}_{\alpha,\beta}$, we can express this as follows: 
	\begin{equation}\label{eq16}
		F_{Y|X}({\rm CoVaR}_{\alpha,\beta} -\tilde{\varepsilon}\,|\, q_\alpha )< \beta <   F_{Y|X}({\rm CoVaR}_{\alpha,\beta} +\tilde{\varepsilon}\,|\, q_\alpha ). 
	\end{equation}
	Let $\varepsilon_1 = \beta-F_{Y|X}({\rm CoVaR}_{\alpha,\beta} -\tilde{\varepsilon}\,|\, q_\alpha )$, $\varepsilon_2 = F_{Y|X}({\rm CoVaR}_{\alpha,\beta} +\tilde{\varepsilon}\,|\, q_\alpha )-\beta$, and $\varepsilon = \min\{\varepsilon_1, \varepsilon_2\}$.
	By Theorem \ref{consistencyofFbar}, there exists $M>0$ such that for $m>M$, the following holds:
	\begin{equation}\label{thm1eqn}
	\pr \left\{ \left| \hat{F}_{n}(y, \hat{q}_\alpha) - F_{Y|X}(y\,|\, q_\alpha ) \right|\geq \varepsilon \right\} \ \to \ 0 \quad \mbox{as $n\to\infty$}, 
	\end{equation}
	 where $y\in \{ {\rm CoVaR}_{\alpha,\beta} -\tilde{\varepsilon},~{\rm CoVaR}_{\alpha,\beta} +\tilde{\varepsilon} \}$. 
	Moreover, if 
	$$
	\big| \hat{F}_{n}({\rm CoVaR}_{\alpha,\beta} - \tilde{\varepsilon}, \hat{q}_\alpha) - F_{Y|X}({\rm CoVaR}_{\alpha,\beta} - \tilde{\varepsilon}\,|\, q_\alpha ) \big|<\varepsilon,
	$$ 
	it implies $\hat{F}_{n}({\rm CoVaR}_{\alpha,\beta} -\tilde{\varepsilon}, \hat{q}_\alpha)< \beta$. Similarly, if
	$$
	\big| \hat{F}_{n}({\rm CoVaR}_{\alpha,\beta} + \tilde{\varepsilon}, \hat{q}_\alpha) - F_{Y|X}({\rm CoVaR}_{\alpha,\beta} + \tilde{\varepsilon}\,|\, q_\alpha ) \big|<\varepsilon,
	$$ 
	it implies $\hat{F}_{n}({\rm CoVaR}_{\alpha,\beta} +\tilde{\varepsilon}, \hat{q}_\alpha)>\beta$. Therefore, when $m>M$, we obtain:  
	\begin{eqnarray}
		\lefteqn{\pr  \left\{ \hat{F}_{n}({\rm CoVaR}_{\alpha,\beta} -\tilde{\varepsilon}, \hat{q}_\alpha)< \beta <   \hat{F}_{n}({\rm CoVaR}_{\alpha,\beta} +\tilde{\varepsilon}, \hat{q}_\alpha) \right\}}\nonumber \\
		&\geq& \pr  \left\{ \left\{ \left| \hat{F}_{n}({\rm CoVaR}_{\alpha,\beta} - \tilde{\varepsilon}, \hat{q}_\alpha) - F_{Y|X}({\rm CoVaR}_{\alpha,\beta} - \tilde{\varepsilon}\,|\, q_\alpha ) \right|<\varepsilon \right\}\right. \nonumber \\
		&&~~~~ \left. \cap \
		\left\{ \left| \hat{F}_{n}({\rm CoVaR}_{\alpha,\beta} + \tilde{\varepsilon}, \hat{q}_\alpha) - F_{Y|X}({\rm CoVaR}_{\alpha,\beta} + \tilde{\varepsilon}\,|\, q_\alpha ) \right|<\varepsilon \right\} \right\} \nonumber \\
		&\geq & 1-\pr \left\{ \left| \hat{F}_{n}({\rm CoVaR}_{\alpha,\beta} - \tilde{\varepsilon}, \hat{q}_\alpha) - F_{Y|X}({\rm CoVaR}_{\alpha,\beta} - \tilde{\varepsilon}\,|\, q_\alpha ) \right|\geq \varepsilon \right\} \nonumber \\
		& &~~~~ -\ \pr \left\{ \left| \hat{F}_{n}({\rm CoVaR}_{\alpha,\beta} + \tilde{\varepsilon}, \hat{q}_\alpha) - F_{Y|X}({\rm CoVaR}_{\alpha,\beta} + \tilde{\varepsilon}\,|\, q_\alpha ) \right|\geq \varepsilon \right\} \label{Bonferroni}\\
		&\to& 1 \quad \mbox{as $n\to\infty$}, \label{limitofEq33}
	\end{eqnarray}
	where Equation \eqref{Bonferroni} follows Bonferroni inequality and Equation \eqref{limitofEq33} follows Equation \eqref{thm1eqn}. Moreover, by Section 1.1.4 of \cite{Serfling1980}, we have
	$$
	\hat{F}_{n}({\rm CoVaR}_{\alpha,\beta} -\tilde{\varepsilon}, \hat{q}_\alpha)< \beta <   \hat{F}_{n}({\rm CoVaR}_{\alpha,\beta} +\tilde{\varepsilon}, \hat{q}_\alpha)
	$$ 
	if and only if 
	$$
	{\rm CoVaR}_{\alpha,\beta} -\tilde{\varepsilon}< \hat{F}^{-1}_{n}(\beta, \hat{q}_\alpha) =\hat{Y} <  {\rm CoVaR}_{\alpha,\beta} +\tilde{\varepsilon}. 
	$$ 
Therefore, we conclude the proof. 
\end{proof}

\section{Proof of Lemma \ref{CRerrorIII}}\label{Lemma4}

\begin{proof}
First, we prove Equation~\eqref{proofofCLT30}. By Taylor's expansion, we have
	\begin{eqnarray}
		\lefteqn{ \E\Big[ \bar{Q}_n(  q_\alpha )-f_X(  q_\alpha ) \Big] } \nonumber \\
		&=& \frac1{\Delta_n} \int_{\R^m} \prod_{j=1}^m K\left( \frac{q_{j,\alpha_j}- x_j}{\delta_{n,j}} \right) f_X(  x )  \dd x  - f_X(  q_\alpha ) \nonumber \\
		&=& \int_{\R^m} \prod_{j=1}^m K(x_j) \Big[ f_X(  q_\alpha + \delta_n \odot x) - f_X(  q_\alpha ) \Big] \dd x \label{Hadamard} \\
		&=& \int_{\R^m} \prod_{j=1}^m K(x_j) \Bigg[ \sum_{|\ell|=1}  \D^\ell f_X(q_\alpha) \cdot (\delta_n \odot x)^\ell + \sum_{|\ell|=2} \frac1{\ell !}  \D^\ell f_X(q_\alpha) \cdot (\delta_n \odot x)^\ell  + \sum_{|\ell|=2} o(1) \delta_n^\ell x^\ell \Bigg] \dd x~~~~ \nonumber \\
		&=& \frac12 \sum_{j=1}^m  \delta_{n,j}^2 \frac{\partial^2 f_X}{\partial x_j^2}(q_\alpha)   \int_{\R} x_j^2K( x_j)\dd x_j  + \sum_{j=1}^m o( \delta_{n,j}^2 ), \label{thm3:eqn4}
	\end{eqnarray}
where $\delta_n \odot x$ in Equation~\eqref{Hadamard} is the element-wise product (also known as Hadamard product), i.e., $\delta_n \odot x=(\delta_{n,1}x_1, \ldots, \delta_{n,m} x_m)^\top$, and Equations~\eqref{Hadamard} and \eqref{thm3:eqn4} hold by Assumption~\ref{assu:kernel1}. 

Similarly, we can prove Equation~\eqref{proofofCLT3}. By Taylor's expansion, we have 
	\begin{eqnarray}
		\lefteqn{ \E\Big[ \bar{R}_n(  q_\alpha ,y_n)-g(  q_\alpha ,y) \Big] } \nonumber \\
		&=& \frac1{\Delta_n} \int_{\R^m} \prod_{j=1}^m K\left( \frac{q_{j,\alpha_j}- x_j}{\delta_{n,j}} \right) g(  x,y_n )  \dd x  - g(q_\alpha,y) \nonumber \\
		&=& \int_{\R^m} \prod_{j=1}^m K(x_j) \Big[ g(  q_\alpha + \delta_n \odot x, y_n) - g(  q_\alpha , y) \Big] \dd x \nonumber \\
		&=& \int_{\R^m} \prod_{j=1}^m K(x_j) \Bigg[ \sum_{|\ell|=1}  \D^\ell g(q_\alpha,y) \cdot (\delta_n \odot x)^\ell + \sum_{|\ell|=2} \frac1{\ell !}  \D^\ell g(q_\alpha,y) \cdot (\delta_n \odot x)^\ell  \nonumber\\ 
		&& ~~~~~~~~~~~~~~~~~~~~ +  f(  q_\alpha , y) \cdot (y_n-y)  + \sum_{|\ell|=2} o(1) \delta_n^\ell x^\ell + o(y_n-y) \Bigg] \dd x \nonumber 
		\\
		&=& \frac12 \sum_{j=1}^m  \delta_{n,j}^2 \frac{\partial^2 g}{\partial x_j^2}(q_\alpha,y)   \int_{\R} x_j^2K( x_j)\dd x_j +  f(  q_\alpha ,y) \cdot (y_n-y) + \sum_{j=1}^m o( \delta_{n,j}^2 ) + o(y_n-y), \nonumber
	\end{eqnarray}
Therefore, we conclude the proof. 
\end{proof}

\section{Proof of Lemma \ref{CRerrorII}}\label{Lemma5}

\begin{proof}
Let Component~(II) be denoted as $\sum_{i=1}^n N_{n,i}$, where 
\begin{align*}
	N_{n,i} \ = \  (n\Delta_n)^{-\frac12} & \Bigg\{ f_X(  q_\alpha )\Big( W_{n,i}(  q_\alpha ) \mathbb{I}\{Y_i\leq y_n\}- \E\big[W_{n,i}(  q_\alpha ) \mathbb{I}\{Y_i\leq y_n \}\big]  \Big)  \\
	& ~~~~~~~~~~~~~~~~~~~~~~~~~~~~~~~~~~~~ -g(  q_\alpha ,y)\Big( W_{n,i}(  q_\alpha ) - \E\big[ W_{n,i}(  q_\alpha )\big]  \Big)  \Bigg\}. 
\end{align*}
By applying Lemma~\ref{lemmaofMV}, we have 
\begin{eqnarray}
\lefteqn{ \lim_{n\to\infty} \frac1{\Delta_n}  \Cov\Big( W_{n,1}(  q_\alpha ) \mathbb{I}\{Y_1 \leq y_n\}, W_{n,1}(  q_\alpha )  \Big) } \nonumber \\
&=& \lim_{n\to\infty} \frac1{\Delta_n} \E\Big[ W_{n,1}^2(q_\alpha) \mathbb{I}\{Y_1\leq y_n \} \Big] - \lim_{n\to\infty} \frac1{\Delta_n}  \left\{ \E\Big[ W_{n,1}(q_\alpha) \mathbb{I}\{Y_1\leq y_n \} \Big]\E\Big[ W_{n,1}(q_\alpha) \Big] \right\} \nonumber \\
&=& g(  q_\alpha , y) \int_{\R^m} \Big[ \prod_{j=1}^m  K ( x_j )  \Big]^2 \dd  x. \label{COV}
\end{eqnarray}
Let $\sigma_n^2= \Var(\sum_{i=1}^n N_{n,i}) $. Then, the asymptotic variance of Component~(II) is given by 
\begin{eqnarray}
\lefteqn{\lim_{n\to\infty} \sigma_n^2  \ =\  \lim_{n\to\infty} n \Var(N_{n, 1})} \nonumber \\
&=&  \lim_{n\to\infty} \frac{1}{\Delta_n}\Bigg[ f_X^2(  q_\alpha ) \Var\Big( W_{n,1}(  q_\alpha ) \mathbb{I}\{Y_1 \leq y_n\}  \Big) + g^2(  q_\alpha ,y) \Var\Big( W_{n,1}(  q_\alpha ) \Big)   \nonumber \\
&&~~~~~~~~~~~~~~~~~~~~~~~~~  - 2 f_X(  q_\alpha )g(  q_\alpha ,y) \Cov\Big( W_{n,1}(  q_\alpha ) \mathbb{I}\{Y_1 \leq y_n\}, W_{n,1}(  q_\alpha )  \Big)  \Bigg] \nonumber \\
&=& \Big[ f_X^2(  q_\alpha ) g(  q_\alpha , y) + f_X(  q_\alpha ) g^2(  q_\alpha ,y) -2 f_X(  q_\alpha ) g^2(  q_\alpha ,y) \Big]  \int_{\R^m} \Big[ \prod_{j=1}^m  K ( x_j )  \Big]^2 \dd x  \ = \ \sigma^2, \label{finiteV}
\end{eqnarray}
where Equation~\eqref{finiteV} holds by applying Lemma~\ref{lemmaofMV} and Equation~\eqref{COV}. 

To prove Lemma~\ref{proofofCLT2}, we can check the Lindeberg-Feller's condition, i.e., for any $\varepsilon>0$, we have 
\begin{equation}
\lim _{n \rightarrow \infty} \frac{1}{\sigma_n^2} \sum_{i=1}^n \E\left[N_{n, i}^2 \cdot \mathbb{I}\left\{\left|N_{n, i}\right| \geq \varepsilon \sigma_n\right\}\right]=0. 
\label{L-Fcondition}
\end{equation}
By Equation \eqref{finiteV}, we can simplify the Lindeberg-Feller's condition to the following expression: 
\begin{eqnarray}
	\lefteqn{\lim _{n \rightarrow \infty} \sum_{i=1}^n \E\left[N_{n, i}^2 \cdot \mathbb{I}\left\{\left|N_{n, i}\right| \geq \varepsilon \sigma_n\right\}\right]} \label{eqn24} \\
	&=&  \lim_{n\to\infty} \frac{1}{\Delta_n} \E\Bigg[ \Bigg\{ f_X(  q_\alpha )\Big( W_{n,i}(  q_\alpha ) \mathbb{I}\{Y_i\leq y_n\}- \E[W_{n,i}(  q_\alpha ) \mathbb{I}\{Y_i\leq y_n \}]  \Big)  \label{eqn240} \\
	&&  ~~~~~~~~~~~~~~~~~~~~~~~~~~~ -g(  q_\alpha ,y)\Big( W_{n,i}(  q_\alpha ) - \E[ W_{n,i}(  q_\alpha )]  \Big)  \Bigg\}^2 \mathbb{I}\left\{\left|N_{n, i}\right| \geq \varepsilon \sigma_n\right\} \Bigg]
	\ =\ 0. \label{eqn2400}
\end{eqnarray}
We can further simplify the above condition by introducing some inequalities. Let's define $E_n$ as follows. 
\begin{align*}
	&E_n \ = \\
	& \E\Big[ \left| f_X(  q_\alpha )\Big( W_{n,i}(  q_\alpha ) \mathbb{I}\{Y_i\leq y_n\}- \E[W_{n,i}(  q_\alpha ) \mathbb{I}\{Y_i\leq y_n \}]  \Big) -g(  q_\alpha ,y)\Big( W_{n,i}(  q_\alpha ) - \E[ W_{n,i}(  q_\alpha )]  \Big)  \right|^{2+\gamma} \Big]^{\frac{1}{2+\gamma}}. 
\end{align*}
Using H\"older's inequality, we can bound the expectation term in Equations~\eqref{eqn240}--\eqref{eqn2400} as follows: for any $\gamma>0$, we have 
\begin{eqnarray}
	\lefteqn{ \E\Bigg[ \Bigg\{ f_X(  q_\alpha )\Big( W_{n,i}(  q_\alpha ) \mathbb{I}\{Y_i\leq y_n\}- \E[W_{n,i}(  q_\alpha ) \mathbb{I}\{Y_i\leq y_n \}]  \Big) } \nonumber \\
	&&  ~~~~~~~~~~~~~~~~~~~ -g(  q_\alpha ,y)\Big( W_{n,i}(  q_\alpha ) - \E[ W_{n,i}(  q_\alpha )]  \Big)  \Bigg\}^2 \mathbb{I}\left\{\left|N_{n, i}\right| \geq \varepsilon \sigma_n\right\} \Bigg] \nonumber \\
	&\leq & \E\Bigg[ \left| f_X(  q_\alpha )\Big( W_{n,i}(  q_\alpha ) \mathbb{I}\{Y_i\leq y_n\}- \E[W_{n,i}(  q_\alpha ) \mathbb{I}\{Y_i\leq y_n \}]  \Big) \right. \nonumber \\
	&& \left. ~~~~~~~~~~~~~~~~~~~ -g(  q_\alpha ,y)\Big( W_{n,i}(  q_\alpha ) - \E[ W_{n,i}(  q_\alpha )]  \Big)  \right|^{2+\gamma} \Bigg]^{\frac{2}{2+\gamma}} \cdot \pr\Big\{ |N_{n,1}|\geq \varepsilon \sigma_n \Big\}^{\frac{\gamma}{2+\gamma}} \nonumber \label{eqn25} \\
	&=&  E_n^2 \cdot \pr\Big\{ |N_{n,1}|\geq \varepsilon \sigma_n \Big\}^{\frac{\gamma}{2+\gamma}}.  \label{eqn250}
\end{eqnarray}
Using Chebyshev's inequality and Equation~\eqref{finiteV}, we can bound the probability term in Equation~\eqref{eqn250} as follows:  
$$
\pr\Big\{ |N_{n,1}|\geq \varepsilon \sigma_n \Big\}^{\frac{\gamma}{2+\gamma}} \ \leq \  \left\{ \frac{\E \big[N_{n,1}^2 \big]}{\varepsilon^2 \sigma_n^2} \right\}^{\frac{\gamma}{2+\gamma}} \ =\ \big( n \varepsilon^2 \big)^{-\frac{\gamma}{2+\gamma}}. 
$$
Using Minkowski's inequality, we can bound the term $E_n$ in Equation~\eqref{eqn250} as follows: 
\begin{align*}
E_n \ \leq \ & f_X(  q_\alpha ) \E\Big[ \left| W_{n,i}(  q_\alpha ) \mathbb{I}\{Y_i\leq y_n\}  \right|^{2+\gamma} \Big]^{\frac{1}{2+\gamma}} \ + \ g(  q_\alpha ,y) \E\Big[ \left| W_{n,i}(  q_\alpha )  \right|^{2+\gamma} \Big]^{\frac{1}{2+\gamma}} \\
& + \ f_X(  q_\alpha ) \E\Big[W_{n,i}(  q_\alpha ) \mathbb{I}\{Y_i\leq y_n\} \Big]  \ + \ g(  q_\alpha , y) \E\Big[W_{n,i}(  q_\alpha )\Big].  
\end{align*}
By Equation~\eqref{lem1:eqn:3} in Lemma \ref{lemmaofMV}, we have  
$$
\lim_{n\to\infty} \frac{1 }{\Delta_n}\E\Big[ \left| W_{n,i}(  q_\alpha ) \mathbb{I}\{Y_i\leq y_n\}  \right|^{2+\gamma} \Big] = g(  q_\alpha ,y) \int_{\R^m} \Big[ \prod_{j=1}^m  K ( x_j )  \Big]^{2+\gamma} \dd x. 
$$
From this result, we can deduce that:
$$
\E \Big[ | W_{n,i}(  q_\alpha ) \mathbb{I}\{Y_i\leq y_n\} |^{2+\gamma} \Big]^{\frac{1}{2+\gamma}} = \mathcal{O} \Big( \Delta_n^{\frac1{2+\gamma}}\Big) \quad \mbox{as $n\to\infty$}. 
$$ 
Similarly, by Lemma~\ref{lemmaofMV}, we have 
$$
\E \Big[ | W_{n,i}(  q_\alpha ) |^{2+\gamma} \Big]^{\frac{1}{2+\gamma}} = \mathcal{O} \Big( \Delta_n^{\frac1{2+\gamma}}\Big) \quad \mbox{as $n\to\infty$}. 
$$
Therefore, we can conclude that: $E_n= \mathcal{O} ( \Delta_n^{1/(2+\gamma)})$ as $n\to\infty$. 

Consequently, as $n$ tends to infinity, the following equality holds: 
$$
\text{Equation~\eqref{eqn24}} = \mathcal{O} \Big( \Delta_n^{\frac{-\gamma}{2+\gamma}}\Big) \cdot \Big( n \varepsilon^2 \Big)^{-\frac{\gamma}{2+\gamma}},
$$
which converges to zero by Assumption~\ref{assu:kernel3}. 
Therefore, the Lindeberg-Feller's condition \eqref{L-Fcondition} holds. By Lindeberg-Feller's Theorem \citep{durrett2019probability}, we conclude the proof. 
\end{proof}

\section{Proof of Lemma \ref{CRerrorI}}\label{Lemma6}

\begin{proof}
	Similar to Equation \eqref{lem2:eqn1}, by Taylor's expansion, we have 
	\begin{eqnarray}
		\lefteqn{ \sqrt{n\Delta_n} \cdot \Big\{\bar{R}_n(   \hat{q}_\alpha  ,y_n) -  \bar{R}_n(   q_\alpha  ,y_n) \Big\} } \label{thm3:eqn1} \\
		& =& \sum_{|\ell|=1}   \sqrt{2 \Delta_n \log\log{n}} \left[  \frac{\sqrt{n}}{\sqrt{2\log\log{n}}}(\hat{q}_\alpha -  q_\alpha) \right]^\ell \frac1{n \Delta_n } \sum_{i=1}^n \D^\ell W_{n, i}(  q_\alpha )  \mathbb{I}\{Y_i\leq y_n \} \label{thm3:1term} \\
		&& ~~~~~~~~~~~~~~~~~~~~~~~~~~~~~~~ + \sum_{|\ell|=1} o_p\left( \sqrt{n\Delta_n} \left(  \hat{q}_\alpha -  q_\alpha \right)^\ell \frac1{n  \Delta_n } \sum_{i=1}^n \D^\ell W_{n, i}(  q_\alpha )  \mathbb{I}\{Y_i\leq y_n \} \right), ~~~~ \label{thm3:2term}
	\end{eqnarray}
	By Lemma \ref{lemmaofMV}, for any $\ell$ such that $|\ell|=1$, we have
	\begin{align}
		\E\left[ \frac{1}{n \Delta_n } \sum_{i=1}^n \D^\ell W_{n, i}(  q_\alpha )  \mathbb{I}\{Y_i\leq y_n \}\right] \ &\to \  \D^\ell g(q_\alpha,y), \label{lem2:var011} \\
		\Var\left(  \frac{1}{n  \Delta_n } \sum_{i=1}^n \D^\ell W_{n, i}(  q_\alpha )  \mathbb{I}\{Y_i\leq y_n \} \right) \ &=\  \mathcal{O} \left( \frac{1}{n\delta_{n}^{2\ell} \Delta_n} \right), \label{lem2:var11} 
	\end{align}
	as $n\to\infty$. By the assumption $n\delta_{n}^{2\ell} \Delta_n \to \infty$ as $n\to\infty$, we have Equation \eqref{lem2:var11} vanishes as $n\to\infty$. 
	Then, by Chebyshev's inequality (Theorem 1.6.4 of \citealt{durrett2019probability}), for $|\ell|=1$, we have 
	\begin{equation}
		\frac{1}{n \Delta_n } \sum_{i=1}^n \D^\ell W_{n, i}(  q_\alpha )  \mathbb{I}\{Y_i\leq y_n \} \ \to \  \D^\ell g(q_\alpha,y), 
		\label{LLN111}
	\end{equation}
	in probability as $n\to\infty$.
	By law of the iterated logarithm for sample quantiles (Section 2.5.1 of \citealt{Serfling1980}), we have 
	$
	[\frac{\sqrt{n}}{\sqrt{2\log\log{n}}}(\hat{q}_\alpha -  q_\alpha)]^\ell
	$
	is bounded almost surely. 
	Then, by the assumption $\Delta_n \log\log{n} \to 0$ as $n\to\infty$, we have Equation \eqref{thm3:1term} converge to zero in probability as $n\to\infty$. Because Equation \eqref{thm3:2term} goes to zero faster than Equation \eqref{thm3:1term} as $n\to\infty$, so Equation \eqref{thm3:eqn1} converges to zero in probability as $n\to\infty$. Similarly, we can prove that $\sqrt{n\delta_n}\{ \bar{Q}_n(  \hat{q}_\alpha ) -\bar{Q}_n(  q_\alpha )\} \to 0$ in probability as $n\to\infty$. Therefore, we conclude the proof. 
\end{proof}

\section{Proof of Theorem \ref{CLT}}\label{Theorem3}

\begin{proof}
As indicated by Equations~\eqref{CLT:results}--\eqref{CLT:eqn10}, we have the following relationship: 
\begin{equation}
	 \pr \left\{ \frac{\sqrt{n\Delta_n}}{\sigma_Y} \big( \hat{Y}-\cvv \big) \leq z \right\}   \ =\ \pr \left\{ \sqrt{n\Delta_n} \Big[\hat{F}_n(y_n,   \hat{q}_\alpha )- \beta  \Big] \geq 0 \right\}, 
\label{APPH:eqn1}
\end{equation}
where $y_n=\cvv + z \sigma_Y(n \Delta_n)^{-1/2}$ for $z\in\R $. 
By recalling that $\hat F_n(y_n,  \hat{q}_\alpha )= \bar{R}_n(  \hat{q}_\alpha ,y_n)/\bar{Q}_n(  \hat{q}_\alpha )$ and $\beta=F_{Y|X}(\cv,  \hat{q}_\alpha )=g(  q_\alpha ,\cv)/f_X(  q_\alpha )$, we obtain the following expression:
\begin{eqnarray}
\lefteqn{ \hat F_n(y_n,  \hat{q}_\alpha ) - \beta} \nonumber \\
&=& \frac{1}{\bar{Q}_n(  \hat{q}_\alpha) f_X(  q_\alpha )}  \Big[ f_X (  q_\alpha ) \Big(\bar{R}_n(  \hat{q}_\alpha , y_n)-g(  q_\alpha , \cv) \Big)-g(  q_\alpha , \cv) \Big(\bar{Q}_n(  \hat{q}_\alpha )-f_X(  q_\alpha ) \Big) \Big]. \nonumber \\
&=& \frac{1 }{\bar{Q}_n(  \hat{q}_\alpha ) f_X(  q_\alpha )}  \Big[ f_X (  q_\alpha ) \Big( \text{Errors~I--1}  + \text{II--1} + \text{III--1} \Big) -g(  q_\alpha , \cv)\Big( \text{Errors~I--2}  + \text{II--2}  + \text{III--2} \Big) \Big], \nonumber
\end{eqnarray}
where 
\begin{alignat*}{2}
	\text{Error~I--1} &\ =\ \bar{R}_n(  \hat{q}_\alpha , y_n) - \bar{R}_n(  q_\alpha , y_n), &\qquad \quad
	 \text{Error~I--2} &\ =\ \bar{Q}_n(  \hat{q}_\alpha ) - \bar{Q}_n(  q_\alpha ), \\
	\text{Error~II--1} &\ = \ \bar{R}_n(  q_\alpha , y_n) - \E[\bar{R}_n(  q_\alpha , y_n)],  &\qquad \quad
	 \text{Error~II--2} &\ = \ \bar{Q}_n(  q_\alpha ) - \E[\bar{Q}_n(  q_\alpha ) ], \\
	\text{Error~III--1} &\ = \ \E[\bar{R}_n(  q_\alpha , y_n)] - g(  q_\alpha , \cv), &\qquad \quad
	\text{Error~III--2} &\ = \ \E[\bar{Q}_n(  q_\alpha ) ] - f_X(  q_\alpha ). 
\end{alignat*}
By virtue of Lemma \ref{CRerrorI}, it is established that both $\sqrt{n\Delta_n}$(Error~I--1) and $\sqrt{n\Delta_n}$(Error~I--2) tend to zero in probability as $n\to\infty$. Let $\bar{\sigma} =\sigma|_{y=\cv}$ as given by Equation~\eqref{sigma}. According to Lemma \ref{CRerrorII}, it follows that
\begin{equation*}
	\frac{\sqrt{n\Delta_n}}{\bar{\sigma}}  \Big\{ f_X(  q_\alpha ) \cdot  (\text{Error~II--1}) -g(  q_\alpha ,\cv)\cdot  (\text{Error~II--2})  \Big\}  \Rightarrow N(0,1) \quad \mbox{as $n\to\infty$}.  
\end{equation*}
Furthermore, Lemma \ref{CRerrorIII} reveals that 
\begin{eqnarray}
	\lefteqn{ \frac{\sqrt{n\Delta_n} }{\bar{\sigma}}  \Big\{ f_X(  q_\alpha ) \cdot  (\text{Error~III--1}) -g(  q_\alpha ,\cv)\cdot  (\text{Error~III--2})  \Big\}  } \label{eqn3000} \\
	&=& \frac{\sqrt{n\Delta_n}}{2\bar{\sigma} } f_X(  q_\alpha ) \int_{\R} t^2K( t)\dd t  \left[ \sum_{j=1}^m   \delta_{n,j}^2 \frac{\partial^2 g}{\partial x_j^2}(q_\alpha,\cv)  \right]  
	\ +\ \frac{ z \sigma_Y}{\bar{\sigma} } f_X(  q_\alpha ) f(  q_\alpha ,\cv)  \label{eqn300} \\
	&&-\ \frac{\sqrt{n\Delta_n} }{2 \bar{\sigma} }  g(  q_\alpha , \cv)   \int_{\R} t^2K( t)\dd t   \left[ \sum_{j=1}^m \delta_{n,j}^2 \frac{\partial^2 f_X}{\partial x_j^2}(q_\alpha) \right] \  +\  \sum_{j=1}^m o( \sqrt{n\Delta_n} \delta_{n,j}^2 ) \ +\  o(1), ~~~~~~ \label{eqn30}
\end{eqnarray}
as $n\to\infty$. It should be noted that, based on the definition of $\sigma_Y$, we have $\bar{\sigma}/\sigma_Y= f_X(  q_\alpha )f(  q_\alpha ,\cv)$. Consequently, the second term in Equation \eqref{eqn300} is equivalent to $z$. Under the assumption $\sqrt{n\Delta_n}\delta_{n,j}^2 \to c_j$ as $n\to\infty$ for certain constants $c_j\geq 0$, $j=1,\ldots,m$, Equation \eqref{eqn3000} converges to a non-zero constant $\mu^\prime+z$ as $n\to\infty$, where
$$
\mu^\prime =  \sum_{j=1}^m  \frac{c_j}{2\bar{\sigma} } \left[ f_X(  q_\alpha ) \frac{\partial^2 g}{\partial x_j^2}(q_\alpha,\cv) 
-  g(  q_\alpha , \cv)  \frac{\partial^2 f_X}{\partial x_j^2}(q_\alpha) \right]  \int_{\R} t^2K( t)\dd t. 
$$

By applying Lemmas~\ref{lemmaofMV}--\ref{2error} and Slutsky's theorem (Section 1.5.4 of \citealt{Serfling1980}), we establish the convergence:  
$$
\sqrt{n\Delta_n} \Big[\hat{F}_n(y_n,   \hat{q}_\alpha )- \beta  \Big] \Rightarrow \frac{\bar{\sigma}}{f_X^2(q_\alpha)} \left[ N(0, 1) + \mu^\prime +z \right] \quad \mbox{as $n\to\infty$}. 
$$
Consequently, based on Equation~\eqref{APPH:eqn1}, we have the following limit:  
\begin{equation}\label{thm3:CLT20}
\lim_{n\to\infty} \pr \left\{ \frac{\sqrt{n\Delta_n}}{\sigma_Y} \big( \hat{Y}-\cvv \big) \leq z \right\}   \ =\ \pr \Big\{ N(0, 1) + \mu^\prime +z \geq 0 \Big\} \ = \  \pr \Big\{  N(-\mu^\prime,1) \leq  z \Big\}, 
\end{equation}
where the last equality holds due to the symmetry of the normal distribution's density function. 
This implies that $\frac{\sqrt{n\Delta_n}}{\sigma_Y} \big( \hat{Y}-\cvv \big) \Rightarrow N(-\mu^\prime,1)$ as $n\to\infty$. 
Using the definition of $\mu_j$, we have
$$
-\frac{ \sqrt{n\Delta_n} }{ \sigma_Y } \sum_{j=1}^m \delta_{n,j}^2 \cdot \mu_j \to \mu^\prime \quad \mbox{as $n\to\infty$}, 
$$
and by combining Equation \eqref{thm3:CLT20} with Slutsky's theorem (Section 1.5.4 of \citealt{Serfling1980}), we obtain the convergence: 
$$
\frac{\sqrt{n\Delta_n}}{\sigma_Y} \Bigg( \hat{Y}-\cv - \sum_{j=1}^m \delta_{n,j}^2 \mu_j \Bigg) \ \Rightarrow \ N(-\mu^\prime,1) +\mu^\prime \ =\ N(0,1)  \quad \mbox{as $n\to\infty$}, 
$$
Therefore, Equation \eqref{thm3:CLT0} is derived, and we conclude the proof. 
\end{proof}

\bibliographystyle{informs2014}  
\bibliography{CoVaR-Kernel}  

\end{document}